\documentclass[11pt]{article}

\usepackage[margin=1in]{geometry}

\usepackage{amsmath}
\usepackage{amssymb}
\usepackage{amsthm}
\usepackage{mathrsfs}
\usepackage{relsize}
\usepackage{bm}
\usepackage{xifthen}
\usepackage{booktabs}
\usepackage{graphicx}
\usepackage{xcolor}

\usepackage[authoryear,round]{natbib}

\usepackage{titlesec}
\titleformat{\section}[block]
  {\filcenter\normalfont\LARGE\bfseries}
  {\thesection.}{0.6em}{}
\titlespacing*{\section}{0pt}{1.8\baselineskip}{0.8\baselineskip}

\usepackage{hyperref}
\hypersetup{
  colorlinks = true,
  linkcolor  = {blue!50!black},
  citecolor  = {blue!50!black},
  urlcolor   = {red!60!black},
  pdfstartview = {Fit}
}

\newcommand{\tmat}[1]{%
	\ifthenelse{\equal{#1}{W}}{\bm{{\widetilde{\mathrm{#1}}}}}%
	{\bm{{\tilde{\mathrm{#1}}}}}%
}

\newcommand{\argmin}{\operatorname*{arg\,min}}

\newtheorem{theorem}{Theorem}
\newtheorem{assumption}{Assumption}

\makeatletter
\newcounter{@remarkctr}

\makeatother

\title{Robust Estimation of Polychoric Correlation for Complex Survey Designs\\
Using Minimum Divergence Methods}
\author{%
  \textbf{Siqi Wei}\\
  Department of Statistics\\
  George Mason University\\
  Fairfax, VA, USA
  \and
  \textbf{David Kepplinger}\\
  Department of Statistics\\
  George Mason University\\
  Fairfax, VA, USA
  \and
  \textbf{Anand N.\ Vidyashankar}\\
  Department of Statistics\\
  George Mason University\\
  Fairfax, VA, USA}
\date{}

\begin{document}
\maketitle

\begin{abstract}

Standard maximum likelihood estimation of polychoric correlations is highly sensitive to contamination in survey data, including response errors, interviewer effects, and careless responding, yet assigns equal weight to all observations regardless of data quality. We develop robust estimators for polychoric correlation under complex survey designs based on two minimum divergence criteria---Hellinger distance (HD) and negative exponential disparity (NED)---incorporating survey weights through Horvitz--Thompson adjusted cell frequencies. For HD, we propose penalized Ridge and Lasso variants that regularize nuisance parameters while leaving the correlation unpenalized, and establish consistency and asymptotic normality with a sandwich covariance reflecting the sampling design. The influence function is finite but not uniformly bounded, reflecting Hellinger's sensitivity to sparse cells. Simulations under Poisson proportional-to-size sampling examine three contamination geometries---concordant upper, concordant lower, and discordant mixed corner---crossed with standard and non-standard latent marginals. The two estimator classes offer complementary advantages: penalized HD methods achieve the lowest mean squared error under concordant contamination, while NED performs best under discordant contamination and under compound misspecification--contamination effects. We provide practical guidelines for method selection based on anticipated contamination patterns in survey practice.
\end{abstract}

\noindent\textit{Keywords}:~robust estimation, Hellinger distance, polychoric correlation, survey weights, complex survey design, contamination.

\bigskip

\section{Introduction}\label{sec:introduction}

Ordinal categorical data are ubiquitous in modern quantitative research, particularly in psychology, education, health sciences, and social surveys.
When researchers measure attitudes, behaviors, or latent constructs using Likert scales, rating scales, or ordered response categories, the observed discrete responses are typically conceptualized as coarse measurements of underlying continuous variables.
The polychoric correlation coefficient formalizes this intuition by estimating the correlation between two latent continuous variables that have been discretized through threshold mechanisms \citep{Pearson1900, Olsson1979}, and serves as a foundational input to structural equation modeling, confirmatory factor analysis, and psychometric scale construction.
Standard maximum likelihood (ML) estimation of polychoric correlations is efficient under correct model specification, but highly sensitive to contamination prevalent in survey data, including response errors, interviewer effects, and careless responding \citep{Krosnick1991, Groves2009, Huang2012}.
In complex survey designs with unequal inclusion probabilities, survey weights can further amplify the influence of contaminated observations \citep{Hulliger1995}, making robust estimation essential.

\citet{WelzMairAlfons2026} recently proposed the E-estimator for the polychoric correlation, focusing on robustness against contamination concentrated in discordant cells such as $(1,5)$ or $(5,1)$, a pattern that can arise from careless responses to reverse-keyed items.
Survey data, however, also exhibit another class of model violations in which uninformative responses cluster in concordant corner cells such as $(1,1)$ or $(5,5)$.
We refer to this pattern as \emph{concordant contamination}, which can be caused by certain behavioral or psychological mechanisms.
For example, extreme response style describes respondents who systematically endorse the most extreme options regardless of item content, producing simultaneous mass at both diagonal extremes \citep{Greenleaf1992, VanVaerenberghThomas2013}.
Acquiescence bias, i.e., the tendency to agree across items, generates upper-corner contamination when reverse-keyed items are absent \citep{Smith1967, BillietMcClendon2000}.
In surveys of bounded variables such as income or wealth, selective non-response and underreporting at one tail distort the marginal extremes \citep{LillardSmithWelch1986, Bollinger1998}.
The direction of contamination has methodological consequences: robust loss functions downweight large Pearson residuals through distinct mechanisms, and an estimator that performs well under discordant contamination need not retain that advantage under concordant contamination.

Minimum divergence estimation offers a principled framework for robust inference across all these contamination patterns.
Within the broad family of divergence measures \citep{Lindsay1994, Basu2011}, Hellinger distance (HD) and the negative exponential disparity (NED) are particularly attractive.
The residual-adjustment function (RAF) of the Hellinger divergence grows at a slower rate (square-root growth) than the Kullback-Leibler divergence associated with the maximum likelihood estimator (linear growth), making minimum HD estimators less affected by contamination than MLE while still achieving full efficiency \citep{Beran1977, Simpson1987}.
NED, on the other hand, has a bounded RAF, and hence it is highly robust against gross contamination without requiring tuning constants.
\citet{KepplingerVidyashankar2026} recently proposed minimum Hellinger distance estimators for continuous superpopulation models under complex survey designs, incorporating Horvitz--Thompson adjusted kernel density estimators as nonparametric plug-ins.

In this paper, we extend divergence-based robust estimation to ordinal survey data, where the weighted empirical cell frequency serves as a natural nonparametric plug-in requiring no bandwidth selection.
We develop minimum HD and NED estimators with penalized Ridge and Lasso variants that regularize nuisance parameters (location and scale) while leaving the correlation parameter $\rho$ unpenalized.

Our contributions are threefold.
First, under correct specification we establish consistency and asymptotic normality of the HD estimator based on Horvitz--Thompson weighted cell 
frequencies, yielding a sandwich covariance whose middle matrix is evaluated under the complex sampling design.
Under fixed contamination or marginal misspecification, the same argument demonstrates convergence to 
the corresponding minimum-divergence pseudo-true parameter rather than to the clean-model parameter.
Second, we characterize how the contrasting residual-adjustment mechanisms of HD and NED interact with contamination geometry and cell sparsity under complex sampling.
HD's sensitivity involves inverse-square-root factors $1/\sqrt{\pi_{ij}}$, which stabilize moderate-probability cells but 
increase sensitivity in very sparse cells.
NED, in contrast, uses a bounded residual-adjustment function applied to the Pearson residual $\delta_{ij}(\boldsymbol{\psi}) = \hat{p}_{ij}^{w}/\pi_{ij}(\boldsymbol{\psi}) - 1$, preventing cells with large empirical excess relative to the fitted model from exerting linearly increasing influence on the estimating equation. 
This distinction drives qualitatively different robustness profiles depending on whether contamination targets concordant or discordant cells in the contingency table.
Third, through comprehensive simulations under Poisson proportional-to-size (PPS) sampling with three contamination geometries (upper, lower, and mixed corner) crossed with two marginal specifications (standard and non-standard), we demonstrate that the estimators offer complementary finite-sample strengths.
Penalized HD methods perform best under concordant contamination with correctly specified marginals---the scenarios introduced in this paper.
NED and the E-estimator are more stable under sparse discordant-cell contamination, corresponding to the setting of \citet{WelzMairAlfons2026}, and under compound misspecification--contamination effects.
We provide practical guidance for method selection based on anticipated contamination patterns in survey practice.
Finally, we apply the estimators to validate self-reported versus measured body mass index in the National Health and Nutrition Examination Survey (NHANES).

\section{Methodology}
\label{sec:methodology}
We now describe the modeling framework for polychoric correlation.
This is followed by our proposed minimum divergence estimators (MDEs).

\subsection{The polychoric model}\label{sec:notation}

We consider a finite population $U_\gamma = \{1, \ldots, N_\gamma\}$ of units whose characteristics of interest $(Y_{1i}, Y_{2i})$ are i.i.d.\ draws from a superpopulation distribution $G$ with joint probability mass function $g$.
Here, $Y_1, Y_2$ denote ordinal categorical variables taking values in $\{1, \ldots, K_l\}, l \in \{1, 2\}$.
For notational simplicity, we assume $K_1=K_2=K$, but the methodology and the following theory are more general.
A sample $S_\gamma \subset U_\gamma$ of size $n_\gamma$ is drawn according to a survey design with first-order inclusion probabilities $P(i \in S_\gamma) > 0$ and corresponding survey weights $w_{\gamma i} = 1/P(i \in S_\gamma)$.
For fixed-size designs, we denote the sampling ratio as $\alpha_\gamma = n_\gamma / N_\gamma < 1$.

The polychoric model posits that observed ordinal responses arise from underlying continuous latent variables $Z_1^*, Z_2^*$ through a thresholding mechanism: for variable $\ell\in\{1,2\}$, $Y_{\ell i}=k$ if $b_{\ell,k-1}<Z_{\ell i}^*\leq b_{\ell,k}$, where $b_{\ell,0}=-\infty$, $b_{\ell,K}=+\infty$, and $b_{\ell,1}<\cdots<b_{\ell,K-1}$ are fixed thresholds.
We follow the common parametric polychoric model in which the latent variables in the superpopulation are bivariate Gaussian, i.e., $(Z_1^*, Z_2^*)^\top \sim N(\boldsymbol{\mu}, \boldsymbol{\Sigma})$ with
\begin{equation*}
\boldsymbol{\mu} = \begin{pmatrix} \theta_1 \\ \theta_2 \end{pmatrix}, \quad
\boldsymbol{\Sigma} = \begin{pmatrix} \sigma_1^2 & \rho\sigma_1\sigma_2 \\ \rho\sigma_1\sigma_2 & \sigma_2^2 \end{pmatrix}.
\end{equation*}

Without any further assumptions, however, the latent process model is not identifiable.
In practice, either the parameters of the latent Gaussian distribution or the discretization thresholds must be fixed, and each avenue is appropriate in different scenarios.
The classical formulation of \citet{Olsson1979} assumes a latent standard bivariate Normal distribution $(\theta_l, \sigma_l) = (0, 1)$ and estimates only the thresholds.
This is most appropriate when the ordinal response categories are subjective gradations of an unobservable continuous attribute, as in Likert-type psychometric items or attitude scales.
Moreover, the number of categories $K$ must not be too large as the dimensions of the parameter space scales with $K$.
An alternative parametrization is where the thresholds are fixed and the parameters of the latent Gaussian distribution are estimated.
This is more appropriate when the response categories are intervals on an externally measurable continuous variable, with cut-points determined by the survey instrument or by administrative criteria (e.g., income ranges).
A detailed comparison, including the empirical contexts under which each is most appropriate, is provided in Appendix~\ref{supp:param}.

In this paper we focus on parametrization where the latent marginal parameters $(\theta_1, \theta_2, \sigma_1, \sigma_2)$ are estimated jointly with the correlation coefficient~$\rho$, while the discretization thresholds $\{b_{\ell,k}\}$ are treated as fixed quantities determined by the study design.
We refer to the parameter vector by $\boldsymbol{\psi} = (\theta_1, \theta_2, \sigma_1, \sigma_2, \rho)^\top \in \Psi \subset \mathbb{R}^5$.
This parametrization is also more amenable to regularization, as introduced in Section~\ref{sec:penalized-variants}, because standard Gaussian marginals are a natural anchor.

Model-predicted cell probabilities are denoted by $\pi_{ij}(\boldsymbol{\psi})$ with $\boldsymbol{\pi}(\boldsymbol{\psi}) = (\pi_{11}(\boldsymbol{\psi}), \ldots, \pi_{KK}(\boldsymbol{\psi}))$. The theoretical probability for cell $(i,j)$ is
\begin{equation*}
\pi_{ij}(\boldsymbol{\psi}) = \Phi_2(b_{1,i}, b_{2,j}; \boldsymbol{\mu}, \boldsymbol{\Sigma})  - \Phi_2(b_{1,i-1}, b_{2,j}; \boldsymbol{\mu}, \boldsymbol{\Sigma})  - \Phi_2(b_{1,i}, b_{2,j-1}; \boldsymbol{\mu}, \boldsymbol{\Sigma}) + \Phi_2(b_{1,i-1}, b_{2,j-1}; \boldsymbol{\mu}, \boldsymbol{\Sigma})
\end{equation*}
where $\Phi_2(\cdot, \cdot; \boldsymbol{\mu}, \boldsymbol{\Sigma})$ denotes the bivariate Gaussian cumulative distribution function.
The Horvitz--Thompson (HT) weighted empirical cell frequencies are
\begin{equation}
\label{eq:ht_freq}
\hat{p}_{ij}^w = \frac{\sum_{k \in S} w_k \, \mathbb{I}(Y_{1k} = i,\, Y_{2k} = j)}{\sum_{k \in S} w_k},
\end{equation}
with indicator function $\mathbb{I}$ and $\hat{\boldsymbol{\pi}}^w = (\hat{p}_{11}^w, \ldots, \hat{p}_{KK}^w)$.

\subsection{Minimum divergence estimation}\label{sec:divergence}

Given the parametric family $\mathcal{F} = \{f_{\boldsymbol{\psi}} : \boldsymbol{\psi} \in \Psi\}$ and a divergence measure $D(f, g)$ between two discrete distributions with probability mass functions $f$ and $g$, the minimum divergence estimator minimizes $D(f_{\boldsymbol{\psi}}, \hat{\boldsymbol{\pi}}^w)$ over $\Psi$.
Different choices of $D$ yield estimators with distinct estimation and robustness properties.
We consider two divergence measures that do not require tuning constants and achieve high asymptotic efficiency under correct specification and high robustness \citep{Beran1977, Simpson1987, Lindsay1994, Basu2011}: the Hellinger distance (HD) and the Negative Exponential Disparity (NED).
Given tables of cell frequencies $\{f_{ij} \colon i,j=1,\dotsc,K\}$ and 
$\{g_{ij} \colon i,j=1,\dotsc,K\}$, these two divergences are defined as
\begin{align*}
D^2_\text{HD}(f, g) &= \frac{1}{2}\sum_{i,j=1}^K \left(\sqrt{f_{ij}} - \sqrt{g_{ij}}\right)^2, \\
D_\text{NED}(f, g) &= \sum_{i,j=1}^{K} g_{ij} \left[\exp\{-\delta_{ij}\} - 1 + \delta_{ij}\right],
\end{align*}
where $\delta_{ij} = f_{ij}/g_{ij} - 1$ is the Pearson residual in cell $(i,j)$. 
The linear term in $D_\text{NED}$ makes the nonnegativity of the criterion 
explicit. Since 
$\sum_{i,j} g_{ij}\,\delta_{ij} = \sum_{i,j} f_{ij} - \sum_{i,j} g_{ij} = 0$, 
this criterion has the same minimizer as the classical form 
$\sum_{i,j} g_{ij}\left[\exp\{-\delta_{ij}\} - 1\right]$.
The minimum Hellinger distance estimator (MHDE) and minimum NED estimator 
(MNEDE) are therefore
\begin{align*}
\hat{\boldsymbol{\psi}}_\text{HD}  &= \argmin_{\boldsymbol{\psi} \in \Psi} 
                                    D_\text{HD}(\hat{\boldsymbol{\pi}}^w, \boldsymbol{\pi}(\boldsymbol{\psi})), \\
\hat{\boldsymbol{\psi}}_\text{NED} &= \argmin_{\boldsymbol{\psi} \in \Psi} 
                                    D_\text{NED}(\hat{\boldsymbol{\pi}}^w, \boldsymbol{\pi}(\boldsymbol{\psi})).
\end{align*}

Large positive values of $\delta_{ij}$ occur when the HT-weighted empirical 
frequency in a cell substantially exceeds the model-implied probability. 
The corresponding residual-adjustment function in the NED estimating 
equation is bounded; equivalently, the exponential factor 
$\exp\{-\delta_{ij}\}$ prevents the cellwise score contribution from growing 
linearly with the residual. NED is therefore more protective against cells 
with observed excess relative to the fitted model, including sparse discordant 
cells under positive latent correlation. This contrasts with HD, where the 
influence of cell frequencies on the parameter estimates scales as 
$1/\sqrt{\pi_{ij}}$, which stabilizes moderate-probability cells but amplifies 
sensitivity in sparse cells.
These estimators are straightforward to compute.
We provide computational details in Appendix~\ref{supp:computation}.

\subsection{Penalized variants}\label{sec:penalized-variants}

For any divergence $D$, the penalized minimum divergence estimator minimizes $D(\hat{p}^w, \pi(\boldsymbol{\psi})) + \lambda P(\boldsymbol{\psi})$, where $P$ penalizes the nuisance parameters $(\theta_1, \theta_2, \sigma_1, \sigma_2)$ towards the standardized model.
The correlation parameter $\rho$---our primary quantity of interest---is not penalized.
The penalty has the added benefit of dampening the influence of pathological outliers on the MHDE by reducing the influence of sparse cell counts.
Here we specifically consider the Ridge and the Lasso penalty, given by
\begin{align*}
P_\text{Ridge}(\boldsymbol{\psi}) &= \theta_1^2 + \theta_2^2 + [\log(\sigma_1)]^2 + [\log(\sigma_2)]^2,\\
P_\text{Lasso}(\boldsymbol{\psi}) &= |\theta_1| + |\theta_2| + |\log(\sigma_1)| + |\log(\sigma_2)|. 
\end{align*}
In our empirical studies in Section~\ref{sec:simulations} we select the regularization parameter $\lambda$ by cross-validation.
Moreover, we apply penalization only to MHDE, since the exponential suppression of large residuals by NED already provides an intrinsic form of regularization against contamination in sparse cells. Although $\rho$ is excluded from the penalty, the joint estimation of 
$\boldsymbol{\psi}$ implies that regularization of the nuisance parameters 
$(\theta_1, \theta_2, \sigma_1, \sigma_2)$ can still influence $\hat\rho$ 
indirectly through their coupling in the divergence objective. We do not 
develop a separate asymptotic theory for the penalized estimator: as 
$\lambda \to 0$ it reduces to the unpenalized MHDE and inherits its 
limiting behavior; for fixed $\lambda > 0$ it instead targets a penalized 
pseudo-true parameter that may differ from the unpenalized minimizer. 
We therefore treat the penalized variants as finite-sample regularization 
devices and evaluate their performance through the simulation study in 
Section~\ref{sec:simulations}.

\subsection{Asymptotic theory}

The asymptotic properties of minimum divergence estimators in the survey-weighted polychoric model follow from standard M-estimation theory. We state the results for MHDE; analogous results hold for MNEDE under the same regularity conditions, with the Hellinger kernel replaced by the negative exponential disparity kernel in the sandwich matrices $\mathbf{A}$ and $\boldsymbol{\Sigma}$.

Focusing on the set of all bivariate Gaussian distributions, $\mathcal F = \{ f_{\boldsymbol{\psi}} \colon \boldsymbol{\psi} \in \boldsymbol{\Psi}  \}$, we impose the following assumptions on the model and estimation framework.

\begin{assumption}[Superpopulation model]\label{assump:superpopulation}
  \((\mathbf{Z}_{\gamma i}, X_{\gamma i})\) are i.i.d.\ across \(i \in U_\gamma\) with $\mathbf{Z}_{\gamma i}\sim g \in L^1(\mathbb R^2)$.
  The observed categorical $Y_{\gamma i}$ arise from the fixed thresholding mechanism defined in Section~\ref{sec:notation}.
  The design may depend on \(X\) but not directly on \(\mathbf{Z}\) given \(X\) (PPS).
\end{assumption}

\begin{assumption}[Design regularity]\label{assump:design-regularity}
  There exists $0 < c_0<\infty$ such that
  $$
    \lim_{\gamma\to\infty}\,
    \Pr\left(
    \max_{i \in U_\gamma} w_{\gamma i}\leq c_0/\alpha_\gamma
    \right)=1.
  $$
  Equivalently, we write $\ \max_i w_{\gamma i}=O_p(1/\alpha_\gamma)$.
  To satisfy this assumption in applications, extremely large inverse inclusion weights can be truncated.
\end{assumption}

\begin{assumption}[Model identifiability]\label{assump:ident}
The parameter $\boldsymbol{\psi}_g$ uniquely minimizes $D(f_{\boldsymbol{\psi}}, g)$ over $\boldsymbol{\psi} \in \Psi$, and for each $\epsilon > 0$, $\inf_{\|\boldsymbol{\psi} - \boldsymbol{\psi}_g\| \geq \epsilon} D(f_{\boldsymbol{\psi}}, g) \geq D(f_{\boldsymbol{\psi}_g}, g) + \Delta(\epsilon)$ for some $\Delta(\epsilon) > 0$.
\end{assumption}

\begin{assumption}[Contamination model for robustness analysis]\label{assump:contam}
For robustness analysis, we consider contaminated superpopulation cell
distributions of the form
\[
g_{\varepsilon} = (1-\varepsilon)\,g_{0} + \varepsilon\, h, 
\qquad 0 \leq \varepsilon < 1,
\]
where $g_{0}$ is the uncontaminated cell distribution and $h$ is an arbitrary
cell distribution on $\{1,\ldots,K\}^{2}$. This assumption is used only to
interpret the population target, robustness properties, and influence-function
behavior under contamination. It is not required for the clean-model 
consistency statement in Theorem~\ref{thm:hd_consist}. Under fixed 
contamination, the corresponding population target is the divergence projection
\[
\boldsymbol{\psi}_{\varepsilon} 
= \argmin_{\boldsymbol{\psi}\in\Psi} D\{f_{\boldsymbol{\psi}}, g_{\varepsilon}\},
\]
which generally need not equal the uncontaminated parameter 
$\boldsymbol{\psi}_{0}$.
\end{assumption}

\begin{assumption}[Smoothness and interiority]\label{assump:smooth}
The probability mass function $f_{\boldsymbol{\psi}}$ is twice continuously differentiable in $\boldsymbol{\psi}$ for each cell $(i,j)$, with derivatives satisfying uniform integrability conditions. The true parameter $\boldsymbol{\psi}_0$ lies in the interior of $\Psi$.
\end{assumption}

We now state the main results for our minimum divergence estimators.

\begin{theorem}[Consistency]\label{thm:hd_consist}
Under Assumptions~\ref{assump:superpopulation}, 
\ref{assump:design-regularity}, \ref{assump:ident}, and 
\ref{assump:smooth}, and under correct specification 
$g = f_{\boldsymbol{\psi}_{0}}$,
\[
\hat{\boldsymbol{\psi}}_{\mathrm{HD}} \xrightarrow{P} \boldsymbol{\psi}_{0}.
\]
\end{theorem}

\begin{proof}
We prove consistency by viewing the estimator as the minimizer of a 
finite-dimensional weighted empirical criterion. The weighted empirical 
Hellinger criterion is
\[
Q_{\gamma}(\boldsymbol{\psi}) 
= \frac{1}{2} \sum_{i,j=1}^{K} 
\left( \sqrt{\hat p_{\gamma,ij}^{w}} 
       - \sqrt{\pi_{ij}(\boldsymbol{\psi})} \right)^{2},
\]
and the corresponding population criterion is
\[
Q(\boldsymbol{\psi}) 
= \frac{1}{2} \sum_{i,j=1}^{K} 
\left( \sqrt{g_{ij}} - \sqrt{\pi_{ij}(\boldsymbol{\psi})} \right)^{2}.
\]
Since $K$ is fixed, it is enough to show
$\max_{1\leq i,j\leq K} | \hat p_{\gamma,ij}^{w} - g_{ij} | \xrightarrow{P} 0$.
The ratio-normalized Horvitz--Thompson cell frequency in~\eqref{eq:ht_freq} is the ratio form of the HT estimator of the population cell proportion.
Equivalently, after dividing numerator and denominator by \(N_\gamma\),
\[
\hat p_{\gamma,ij}^{w}
=
\frac{
N_\gamma^{-1}\sum_{k\in U_\gamma}
\delta_{\gamma k}w_{\gamma k}
\mathbf 1\{Y_{1k}=i,Y_{2k}=j\}
}{
N_\gamma^{-1}\sum_{k\in U_\gamma}
\delta_{\gamma k}w_{\gamma k}
}.
\]
Under Assumptions~\ref{assump:superpopulation} and
\ref{assump:design-regularity}, the normalized numerator converges in
probability to \(g_{ij}\), the normalized denominator converges in probability
to \(1\), and therefore
$\hat p_{\gamma,ij}^{w} \xrightarrow{P} g_{ij}$
for each cell $(i,j)$. 
Because the number of cells $K^{2}$ is fixed,
$\max_{1\leq i,j\leq K} | \hat p_{\gamma,ij}^{w} - g_{ij} | \xrightarrow{P} 0$.
The map $p \mapsto \sqrt{p}$ is continuous on $[0,1]$, and the criterion is 
a finite sum. Hence
$\sup_{\boldsymbol{\psi}\in\Psi} 
| Q_{\gamma}(\boldsymbol{\psi}) - Q(\boldsymbol{\psi}) | \xrightarrow{P} 0$.
By the separation condition in Assumption~\ref{assump:ident}, the population criterion $Q$ has a unique minimizer, denoted by $\boldsymbol{\psi}_{g} = \argmin_{\boldsymbol{\psi}\in\Psi} Q(\boldsymbol{\psi})$.
The standard argmin theorem then gives
$\hat{\boldsymbol{\psi}}_{\mathrm{HD}} 
= \argmin_{\boldsymbol{\psi}\in\Psi}Q_{\gamma}(\boldsymbol{\psi})
\xrightarrow{P} \boldsymbol{\psi}_{g}$.
Under the assumed correct model specification there exists a $\bm\psi_0$ such that $g = f_{\boldsymbol{\psi}_{0}}$.
Combined with the identifiability condition in Assumption~\ref{assump:ident} this yields $\hat{\boldsymbol{\psi}}_{\mathrm{HD}} \xrightarrow{P} \boldsymbol{\psi}_{0}$.
\end{proof}

Before stating the asymptotic normality result, we discuss the design-specific ingredient needed for the finite-dimensional ordinal model.
Let
\[
\hat{\mathbf p}_{\gamma}^{w} = (\hat p_{\gamma,11}^{w},\ldots,\hat p_{\gamma,KK}^{w})^\top,
\qquad
\mathbf g = (g_{11},\ldots,g_{KK})^\top .
\]
Assume that, under the joint superpopulation and sampling-design distribution,
\[
\sqrt{n_\gamma}\,(\hat{\mathbf p}_{\gamma}^{w} - \mathbf g) 
\Rightarrow N(\mathbf 0, \mathbf V_d),
\]
where $\mathbf V_d$ is the limiting design-adjusted covariance matrix of the
ratio-normalized Horvitz--Thompson cell-frequency vector. For SRS-WOR,
$\mathbf V_d$ contains the usual finite-population correction; for Poisson-PPS
or other unequal-probability designs, $\mathbf V_d$ is the corresponding
Horvitz--Thompson design covariance.

\begin{theorem}[Asymptotic normality]\label{thm:hd_asymptotic}
Assume the conditions of Theorem~\ref{thm:hd_consist}. In addition, suppose
that the population minimizer $\bm{\psi}_g$ lies in the interior of
$\Psi$, that every cell has positive population probability $g_{ij} > 0$, and
that the map $\bm{\psi} \mapsto \bm{\pi}(\bm{\psi})$ 
is twice continuously differentiable in a neighborhood of $\bm{\psi}_g$.
Define the Hellinger estimating equation
$\mathbf U(\bm{\psi}, \mathbf p) = \nabla_{\bm{\psi}} D_{\mathrm{HD}}(\bm{\pi}(\bm{\psi}), \mathbf p)$
and assume
\[
\mathbf A 
= \nabla_{\bm{\psi}} \mathbf U(\bm{\psi}, \mathbf g)\Big|_{\bm{\psi} = \bm{\psi}_g}
\]
is nonsingular, and that
$\sqrt{n_\gamma}\,(\hat{\mathbf p}_{\gamma}^{w} - \mathbf g) 
\Rightarrow N(\mathbf 0, \mathbf V_d)$.
Let
$\mathbf J_g 
= \left.\nabla_{\mathbf p} \mathbf U(\boldsymbol{\psi}_g, \mathbf p)\right|_{\mathbf p = \mathbf g}$
denote the Jacobian of the Hellinger estimating equation with respect to the cell-probability vector.
Then
\[
\sqrt{n_\gamma}\,(\hat{\bm{\psi}}_{\mathrm{HD}} - \boldsymbol{\psi}_g) 
\Rightarrow N\!\left(\mathbf 0, \mathbf A^{-1} \mathbf B_d \mathbf A^{-\top}\right),
\quad \text{where} \quad
\mathbf B_d = \mathbf J_g \mathbf V_d \mathbf J_g^\top.
\]
Under correct specification, $\mathbf g = \boldsymbol{\pi}(\boldsymbol{\psi}_0)$
and $\boldsymbol{\psi}_g = \boldsymbol{\psi}_0$, so the same result holds
centered at $\boldsymbol{\psi}_0$.
\end{theorem}

\begin{proof}
The estimator satisfies 
$\mathbf U(\hat{\boldsymbol{\psi}}_{\mathrm{HD}}, \hat{\mathbf p}_{\gamma}^{w}) = \mathbf 0$.
A first-order Taylor expansion around $(\boldsymbol{\psi}_g, \mathbf g)$ gives
\[
\mathbf 0 = \mathbf U(\boldsymbol{\psi}_g, \mathbf g)
+ \mathbf A\,(\hat{\boldsymbol{\psi}}_{\mathrm{HD}} - \boldsymbol{\psi}_g)
+ \mathbf J_g\,(\hat{\mathbf p}_{\gamma}^{w} - \mathbf g)
+ o_p(n_\gamma^{-1/2}).
\]
Since $\boldsymbol{\psi}_g$ minimizes the population criterion,
$\mathbf U(\boldsymbol{\psi}_g, \mathbf g) = \mathbf 0$. Therefore
\[
\sqrt{n_\gamma}\,(\hat{\boldsymbol{\psi}}_{\mathrm{HD}} - \boldsymbol{\psi}_g)
= -\mathbf A^{-1}\mathbf J_g\,\sqrt{n_\gamma}\,(\hat{\mathbf p}_{\gamma}^{w} - \mathbf g) + o_p(1).
\]
The assumed CLT for the ratio-normalized HT cell-frequency vector and Slutsky's
theorem yield
\[
\sqrt{n_\gamma}\,(\hat{\boldsymbol{\psi}}_{\mathrm{HD}} - \boldsymbol{\psi}_g) 
\Rightarrow N\!\left(\mathbf 0, \mathbf A^{-1}\mathbf J_g\mathbf V_d\mathbf J_g^\top\mathbf A^{-\top}\right).
\qedhere
\]
\end{proof}

\begin{theorem}[Influence function of the MHDE functional]\label{thm:hd_if}
Under Assumptions~\ref{assump:superpopulation}--\ref{assump:smooth}, suppose
that the population minimizer $\boldsymbol{\psi}_g$ is interior, the Hessian
\[
\mathbf{A} = \nabla_{\boldsymbol{\psi}}^{2} 
D_{\mathrm{HD}}(\boldsymbol{\pi}(\boldsymbol{\psi}), \mathbf{g})\Big|_{\boldsymbol{\psi} = \boldsymbol{\psi}_g}
\]
is nonsingular, and $g_{ij} > 0$ and $\pi_{ij}(\boldsymbol{\psi}_g)>0$ for all cells.
Then the influence function of the MHDE functional 
\[
T(G) = \argmin_{\boldsymbol{\psi} \in \Psi} D_{\mathrm{HD}}(\boldsymbol{\pi}(\boldsymbol{\psi}), \mathbf{g})
\]

at the cell $(i, j)$ is
\[
\mathrm{IF}(i, j; T, G) = -\mathbf{A}^{-1} \boldsymbol{\phi}(i, j),
\quad 
\]
where
\[
\quad
\boldsymbol{\phi}(i, j) = \nabla_{\mathbf{p},\boldsymbol{\psi}} 
D_{\mathrm{HD}}(\boldsymbol{\pi}(\boldsymbol{\psi}), \mathbf{p})\Big|_{(\boldsymbol{\psi}, \mathbf{p}) = (\boldsymbol{\psi}_g, \mathbf{g})}
(\mathbf{e}_{ij} - \mathbf{g}).
\]
Equivalently, the influence function is obtained by differentiating the population estimating equation under the contaminated distribution
$\mathbf{g}_{\varepsilon} = (1 - \varepsilon)\mathbf{g} + \varepsilon \mathbf{e}_{ij}$.
For fixed $K$ and strictly positive cell probabilities, this influence function is finite at every cell.
However, its magnitude depends on inverse cell-probability factors and need not be uniformly bounded as some $g_{ij}$ or 
$\pi_{ij}(\boldsymbol{\psi}_g)$ approach zero.
\end{theorem}

\begin{proof}
The result follows by differentiating the population estimating equation
\[
\mathbf U(\boldsymbol{\psi},\mathbf p)
=
\nabla_{\boldsymbol{\psi}}
D_{\mathrm{HD}}(\boldsymbol{\pi}(\boldsymbol{\psi}),\mathbf p)
\]
along the contamination path
\[
\mathbf g_{\varepsilon}
=
(1-\varepsilon)\mathbf g+\varepsilon \mathbf e_{ij}.
\]
At \(\varepsilon=0\), the implicit differentiation of
\(\mathbf U(T(G_{\varepsilon}),\mathbf g_{\varepsilon})=\mathbf 0\) gives
\[
\operatorname{IF}(i,j;T,G)
=
-\mathbf A^{-1}\boldsymbol{\phi}(i,j).
\]
This is the standard Gateaux-derivative calculation for an implicitly defined statistical functional; see \citet{hampel_robust_1986}.
The detailed derivation and the discussion of finite-support versus uniform boundedness are given in Appendix~\ref{supp:proofs}.
\end{proof}

\begin{remark}[Efficiency and Hellinger geometry]\label{rem:hd_vs_ml}
Under correct parametric specification, the MHDE has the same first-order
asymptotic variance as the MLE.
Indeed, the Hellinger estimating equation can be written as
\[
\bm 0 = \sum_{i,j} \sqrt{p_{ij}\,\pi_{ij}(\boldsymbol{\psi})}\,
\nabla_{\boldsymbol{\psi}} \log \pi_{ij}(\boldsymbol{\psi}),
\]
where $p_{ij}$ denotes the population cell probability, or its HT-weighted empirical version in the sample criterion.
Thus the HD estimating equation uses the usual likelihood score contribution $\nabla_{\boldsymbol{\psi}} \log \pi_{ij}(\boldsymbol{\psi})$, but weights it by
the square-root of the empirical and model-implied cell probabilities.

At the correctly specified model, $p_{ij} = \pi_{ij}(\boldsymbol{\psi}_0)$, and the weight reduces to $\pi_{ij}(\boldsymbol{\psi}_0)$.
Consequently, the first-order linearization of the MHDE agrees with that of ML, giving the same first-order influence function and asymptotic variance.

Away from the model, however, the cross-term $\sqrt{p_{ij}\,\pi_{ij}(\boldsymbol{\psi})}$ captures the trade-off between the nonparametric cell distribution and the parametric model.
A cell contributes substantially only when both the empirical distribution and the fitted model assign it non-negligible mass.
This square-root geometry explains why HD can be less sensitive than ML to moderate rare-cell perturbations.
However, it does not yield uniform boundedness: when contamination is concentrated in extremely
sparse cells, the score term
\[
\nabla_{\boldsymbol{\psi}} \log \pi_{ij}(\boldsymbol{\psi}) 
= \nabla_{\boldsymbol{\psi}} \pi_{ij}(\boldsymbol{\psi}) / \pi_{ij}(\boldsymbol{\psi})
\]

can still be large.

Finally, the estimator functional is defined at the superpopulation level. 
The sampling design affects the estimator's variance through the HT-weighted empirical cell-frequency process, not the population functional itself.
\end{remark}

\section{Simulation study}\label{sec:simulations}
We conduct comprehensive Monte Carlo simulations to evaluate the performance of the proposed robust estimators under complex survey designs with contamination.
We compare the proposed divergence-based estimators (MHDE, Ridge MHDE, Lasso MHDE, and MNEDE), against the non-robust baseline maximum likelihood (ML) estimator which maximizes $\sum_{i,j} \hat{p}_{ij}^w \log \pi_{ij}(\boldsymbol{\psi})$.

We further compare against the E-estimator (EE) \citep{WelzMairAlfons2026}, which downweights cells with extreme Pearson residuals via a discrepancy function with tuning constant $c$.
We follow the advice in \citet{WelzMairAlfons2026} and set $c = 0.6$ and compute the estimator with the \texttt{robcat} R package \citep{robcat2026}.

For the penalized variants of MHDE, we select $\lambda$ using 5-fold cross-validation on median variation loss $\text{MV}(\hat{\boldsymbol{\pi}}^w, \boldsymbol{\pi}(\boldsymbol{\psi})) = \frac{1}{2} \text{Med}_{i,j} |\hat{p}_{ij}^w - \pi_{ij}(\boldsymbol{\psi})|$, evaluating $\lambda \in [0, 0.5]$ over 15 equally-spaced values.

\begin{remark}[Design-based versus i.i.d.-based inference]\label{rem:design_vs_iid}
The EE is designed for i.i.d.\ samples and not for survey data.
We circumvent this mismatch by supplying the HT-weighted contingency table (rounded to the nearest integers) to maintain design-awareness.
Although the resulting point estimate is consistent (when the rounding is ignored), the i.i.d.\ variance formula for EE systematically underestimates the design-based variance whenever the sampling weights are variable.
The asymptotic covariance for the MHDE and MNEDE in Theorem~\ref{thm:hd_asymptotic}, on the other hand, takes the sampling design into account as the inclusion probabilities $\Pr(i \in S_{\gamma})$ enter the weighted score function.
The magnitude of the discrepancy is captured by the design effect $\mathrm{DEFF} = n / n_\mathrm{eff}$, with $n_\mathrm{eff} = (\sum_i w_{\gamma i})^2 / \sum_i w_{\gamma i}^2$.
A detailed comparison of the two variance forms is provided in Appendix~\ref{supp:weighted_ee}.
\end{remark}

\paragraph{Survey design and superpopulation framework.}
Each Monte Carlo replication proceeds in three stages.
First, a finite population of $N = 5{,}000$ units is generated from the superpopulation: for each unit $i$, we draw \((X_i,Z_{1i}^*,Z_{2i}^*)\) from a trivariate normal distribution,
where \(X_i\in\mathbb R\) is an auxiliary variable with $\text{cor}(X_i, Z_{1i}^*) = \rho_{XZ} = 0.25$, and the latent variables $(Z_{1i}^*, Z_{2i}^*)$ have the specified marginal and correlation parameters. The positive PPS size measure is taken to be $M_i=\exp(X_i)$.

The latent variables are then discretized to ordinal responses $(Y_{1i}, Y_{2i})$ via thresholds. Second, a sample is drawn by Poisson-PPS sampling with inclusion probabilities $P(i \in S) = n M_i / \sum_j M_j$ and target sample size $n = 500$.
We cap each inclusion probability at one, so that the $\pi_i$ remain valid probabilities and units with $\pi_i = 1$ enter the sample with certainty.
Third, contamination is applied at the unit level (see below), and weighted cell frequencies $\hat{p}_{ij}^w$ are computed.

\paragraph{Parameter configurations.}
We examine two scenarios: standard and non-standard.
In the standard scenario, the true parameters are $\theta_1 = \theta_2 = 0$, $\sigma_1 = \sigma_2 = 1$, and $\rho = 0.5$.
In the non-standard scenario, the true parameters are $\theta_1 = \theta_2 = 0.5$, $\sigma_1 = \sigma_2 = 0.8$, and $\rho = 0.5$, but the thresholding mechanism is the same as in the standard scenario.
This location shift and scale compression causes category boundaries to misalign with the true latent distribution, a situation common in survey data with ceiling/floor effects.

The cell probability structure under the true model is essential for understanding estimator behavior.
Figure~\ref{fig:cell_sparsity} shows the cell probabilities $\pi_{ij}$ under $\rho = 0.5$ and standard marginals with the chosen boundaries.
Concordant corners (cells $(1,1)$ and $(5,5)$) have $\pi \approx 0.012$, while discordant corners (cells $(1,5)$ and $(5,1)$) have $\pi \approx 6 \times 10^{-5}$---a 200-fold sparsity difference induced by the positive correlation.
In the scenarios with non-standard marginals this imbalance is further amplified, particularly in the bottom left corner cell $(1, 1)$.

Common parameters across simulations are $K = 5$ categories with cell boundaries $b_k = \Phi^{-1}(q_k)$ where $q_k \in \{0, 0.05, 0.20, 0.80, 0.95, 1\}$ yielding expected proportions under standard marginals [5\%, 15\%, 60\%, 15\%, 5\%], and 100 Monte Carlo replications per contamination level.

\paragraph{Contamination mechanism.}
Contamination is applied at the unit level after sampling. At each contamination level $\varepsilon \in \{0\%, 5\%, 10\%, 15\%, 20\%\}$, a proportion $\varepsilon$ of sampled units are selected uniformly at random and have their $(Y_1, Y_2)$ replaced with values in the target corner region.
Of the contaminated units, 80\% are placed in the main corner cell and 20\% are diffused uniformly to the three immediately adjacent cells. We examine three corner types: upper $(K, K)$, lower $(1, 1)$, and mixed $(1, K)$. Upper and lower corners are concordant cells under positive $\rho$, so contamination there reinforces the correlation direction (concordant contamination); the mixed setting targets a discordant cell, pushing against the true correlation (discordant contamination). The mixed-corner case corresponds to the contamination structure studied in \citet{WelzMairAlfons2026}, which models careless responses to reverse-keyed items; the lower- and upper-corner cases extend this framework to acquiescence and extreme response style scenarios discussed in Section~\ref{sec:introduction}. Because contamination operates at the unit level, each contaminated observation carries its original survey weight into the weighted cell frequencies.
Therefore, contaminated units with large weights (low inclusion probabilities) exert disproportionate influence on $\hat{p}_{ij}^w$.
The highlighted cells in Figure~\ref{fig:cell_sparsity} indicates the contamination targets in the three contamination mechanisms.

\paragraph{Performance metrics.}
For each parameter $\eta \in \{\rho, \theta_1, \sigma_1\}$ and estimator, we calculate bias $\text{Bias}(\hat{\eta}) = R^{-1} \sum_{r=1}^{R} (\hat{\eta}_r - \eta_0)$ and mean squared error $\text{MSE}(\hat{\eta}) = R^{-1} \sum_{r=1}^{R} (\hat{\eta}_r - \eta_0)^2$, where $R = 100$.

\begin{figure}[htbp]
    \centering
    \includegraphics[width=0.65\textwidth]{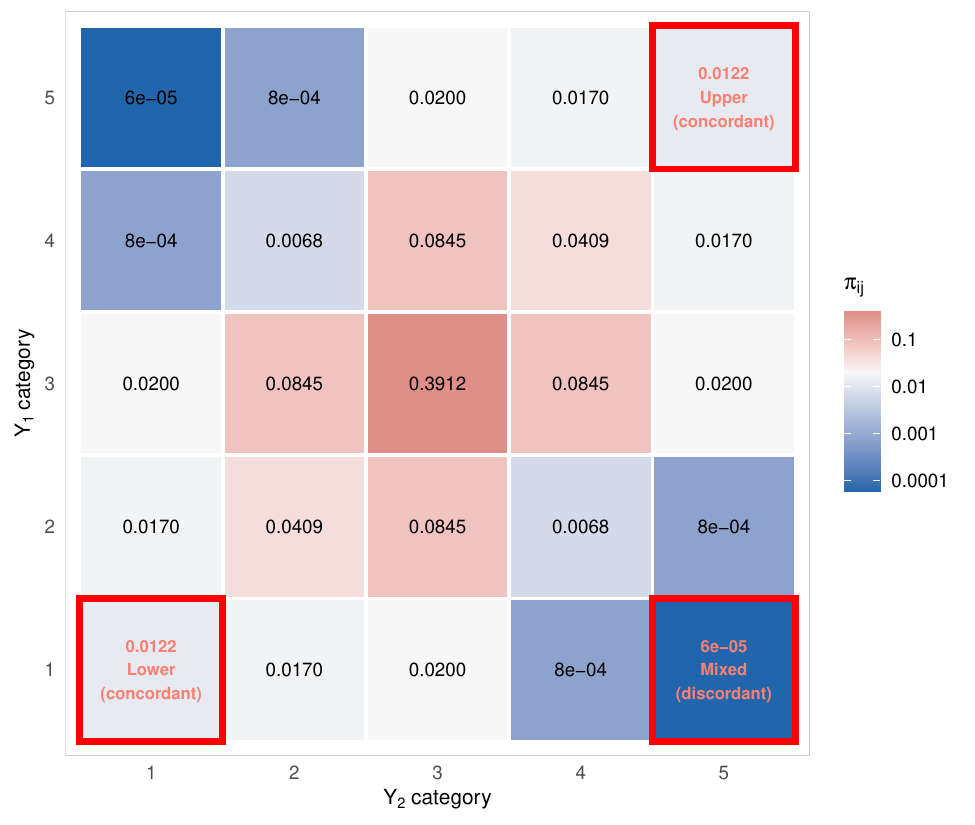}
    \caption{Cell probabilities $\pi_{ij}$ under $\rho = 0.5$. Cells with red borders indicate contamination targets.}
    \label{fig:cell_sparsity}
\end{figure}

\subsection{Results}

\begin{figure}[htbp]
    \centering
    \includegraphics[width=\textwidth]{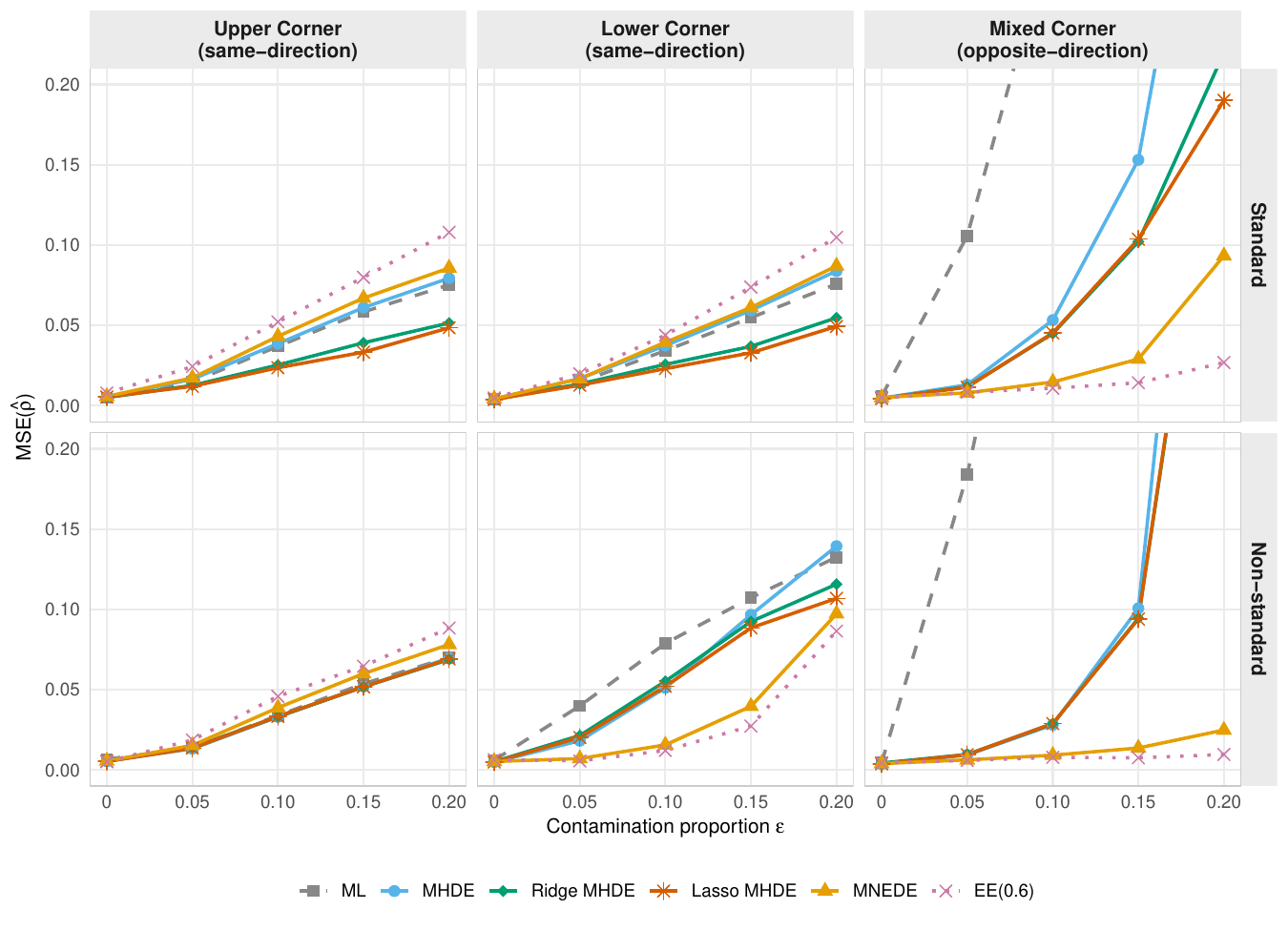}
    \caption{MSE of $\hat{\rho}$ across contamination scenarios under PPS sampling ($N = 5{,}000$, $n \approx 500$, 100 replications). The rows show standard vs.\ non-standard marginals, whereas the columns separate upper, lower (concordant), and mixed (discordant) corner contamination.}
    \label{fig:mse_rho_panel}
\end{figure}

\begin{figure}[htbp]
    \centering
    \includegraphics[width=\textwidth]{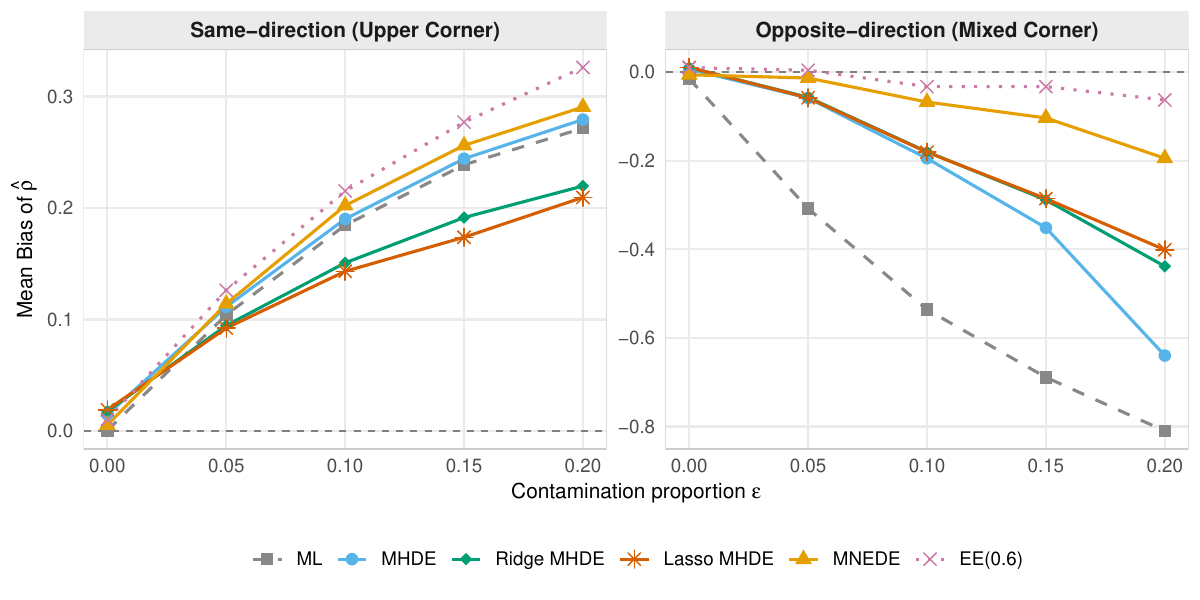}
    \caption{Mean bias of $\hat{\rho}$ under standard marginals. Left: concordant (upper corner) contamination, where all estimators are biased upward but penalized MHDE variants (Lasso MHDE, Ridge MHDE) achieve the smallest bias. Right: discordant (mixed corner) contamination, where MNEDE maintains near-zero bias while HD-based methods and ML degrade severely.}
    \label{fig:bias_rho_key}
\end{figure}

As predicted by our theory, with perfectly clean data ($\varepsilon = 0$) all estimators achieve practically identical results with similar MSE (Figure~\ref{fig:mse_rho_panel}), and the bias of the robust estimators is only marginally higher than the bias of the baseline MLE (Figure~\ref{fig:bias_rho_key}).
The penalization does not introduce excessive bias since the parameter of interest $\rho$ is excluded from regularization.
Once contamination is considered, however, the performance paths diverge based entirely on the geometry of the outliers.

Under concordant contamination, mimicking behaviors like acquiescence bias or extreme response styles, the penalized MHDE variants are substantially more stable than other estimators, provided the marginal are standard normal.
As shown in the left columns of Figures~\ref{fig:mse_rho_panel} and~\ref{fig:bias_rho_key}, they maintain low MSE and bias for $\rho$.
They also tightly control the MSE for nuisance parameters ($\theta_1$ and $\sigma_1$) as shown in the additional simulation results in Appendix~\ref{supp:nuisance}.
The unpenalized MHDE and MNEDE behave similarly to the non-robust MLE in this scenario because the Pearson residual remains relatively small under concordant contamination.
The EE's downweighting mechanism, on the other hand, is too aggressive for the concordant contamination scenario, rendering the EE substantially worse than the other estimators.
A larger value for the tuning constant $c$ alleviates some of those shortcomings of the EE in this contamination scenario, but this substantially deteriorates the EE in the other settings.

Under non-standard marginals, the lower corner contamination setup changes the dynamic substantially (bottom center panel in Figure~\ref{fig:mse_rho_panel}).
Here, the misalignment between the fixed thresholds and marginal distributions shrinks the model-implied cell probability ($\pi_{11}$).
This smaller probability drastically inflates the Hellinger distance's $1/\sqrt{\pi_{11}}$ scaling factor, acting as a sensitivity amplifier that causes HD-based estimators to degrade heavily under this contamination.
MNEDE and EE are less affected by this inverse-square-root sensitivity because their estimating equations use bounded residual adjustments rather than the HD's square-root scaling, allowing them to remain robust and substantially outperform the rest.

This contrast becomes even more explicit under discordant contamination shown in the right column of Figure~\ref{fig:mse_rho_panel}.
In this regime, often encountered due to careless responding to reverse-keyed items, the contaminated cells are incredibly sparse under a positive $\rho$.
Because Hellinger distance only offers polynomial dampening, HD methods suffer severe negative bias, almost as bad as the MLE, as shown in the right columns of Figures~\ref{fig:mse_rho_panel} and~\ref{fig:bias_rho_key}.
MNEDE’s bounded negative-exponential residual adjustment and EE's sharp-cutoff downweighting prove highly resilient to this contamination, allowing both to maintain small bias and MSE for $\varepsilon \leq 0.15$.
As the contamination becomes more prevalent, however, the bias of MNEDE starts to increase, while the EE remains relatively stable.

\begin{remark}[Differences in parametrization]\label{rem:ee-parametrization}
The EE shown here assumes standardized latent variables ($\theta_l = 0$, $\sigma_l = 1$) and estimates $\rho$ alongside the threshold parameters; we therefore include EE only in the comparison of $\hat{\rho}$.
In Appendix~\ref{supp:welz_empirical} we apply the EE with the alternative parametrization to match the proposed divergence-based estimators.
These results show that the EE performs poorly with the alternative parametrization, primarily because a fixed tuning constant $c$ does not work equally well for all $\boldsymbol{\psi} \in \boldsymbol{\Psi}$.
The tuning-free MHDE and MNEDE, on the other hand, perform equally well in both parametrizations, but penalization is meaningful only in the location/scale parametrization used here.
\end{remark}

\section{Application: validating self-reported against measured BMI}
\label{sec:application}

We illustrate the estimators on a measurement-validation problem in which two ordinal variables are designed to capture the same latent construct, hence the latent correlation should be close to one. 
Departures from unity are therefore attributable to reporting error rather than to a genuinely weak association, which makes the setting a clean test of robustness: a well-behaved estimator should recover a correlation near the continuous-scale benchmark, while a non-robust estimator should be pulled toward the off-diagonal mass induced by misreporting.

The data come from the 2021--2023 National Health and Nutrition Examination Survey (NHANES).
The measured BMI is the examiner-recorded body mass index; the self-reported BMI is derived from interview-reported height and weight.
Both are discretized at the World Health Organization cutpoints $18.5, 25, 30, 35\ \mathrm{kg/m^2}$, yielding $K=5$ ordered categories from underweight to obese~II$+$.
We use the mobile-examination-center survey weights included in the dataset.

Self-reported anthropometry systematically underestimates weight, and the resulting BMI misclassification is well documented \citep{ConnorGorber2007}. 
Under-reporting is linked to social-desirability pressures and body-image concerns, and is more pronounced among women and younger adults \citep{Larson2000, King2018}. 
We therefore restrict our attention to the $18$--$22$ age range ($n = 429$). 
Because the subset is a subpopulation of the full survey, point estimation uses the subset-restricted weighted frequencies.

The reporting error has the directional, locally concentrated form that motivates the joint parameterization. 
On the continuous scale the mean self-minus-measured shift is $-0.51\ \mathrm{kg/m^2}$, indicating systematic under-reporting; on the categorical scale over $75\%$ of respondents fall on the diagonal, while the off-diagonal mass is asymmetric, with nearly $15\%$ reporting a \emph{lower} category than measured against only $7\%$ reporting a higher one.
The weighted cell proportions (Figure~\ref{fig:bmi}a) place this discordance in cells adjacent to the WHO cutpoints rather than spreading it uniformly.

\subsection{Results}
\label{subsec:bmi-results}

On the continuous scale the design-weighted Pearson correlation between self-reported and measured BMI is $0.999$, confirming that the two instruments track the same latent quantity almost perfectly. 
Compared to this benchmark, the maximum-likelihood polychoric estimate is $\hat\rho_{\mathrm{ML}} = 0.932$, and the weighted Spearman correlation is $0.935$: both are pulled downward by the sparse off-diagonal misclassifications.

The robust estimators, on the other hand, recover the strong latent association (Table~\ref{tab:bmi}, Figure~\ref{fig:bmi}c).
MHDE-based and MNEDE agree on a polychoric correlation of approximately $0.973$.
The E-estimator under the classical parametrization yields a slightly lower correlation of $0.96$.
Importantly, a survey replicate-weighted bootstrap shows high variability for the ML and the EE, whereas the MHDE variants and MNEDE appear relatively stable.
Applying the E-estimator in our parametrization does not work well in this application with a point estimate close to the ML and extremely wide confidence intervals.

The Pearson-residual diagnostic for the Lasso MHDE fit (Figure~\ref{fig:bmi}b) makes the mechanism explicit. 
The fitted model leaves one cell --- self-reported overweight against measured underweight, a logically near-impossible combination carrying only $0.007$ weighted mass --- with an extreme standardized residual, while all substantively plausible cells are fitted well. 
The robust criteria down-weight this contradictory cell rather than letting it determine $\hat\rho$, which is precisely why the MHDE and MNEDE are close to the continuous-scale benchmark while ML deviates substantially. 

\begin{table}[t]
\centering
\caption{Polychoric correlation between self-reported and measured BMI
(NHANES 2021--2023, ages 18--22, $n=429$). Confidence intervals are design-based
$95\%$ intervals taken as the 2.5th and 97.5th percentiles of the
replicate-bootstrap distribution.}
\label{tab:bmi}
\begin{tabular}{lccc}
\toprule
Estimator & $\hat\rho$ & 95\% CI \\
\midrule
ML            & 0.932 & [ 0.891   0.973 ] \\
MHDE          & 0.973 & [ 0.963   0.990 ]  \\
Ridge MHDE    & 0.974 & [ 0.963   0.990 ] \\
Lasso MHDE    & 0.972 & [ 0.962   0.988 ] \\
MNEDE         & 0.975 & [ 0.960   0.987 ] \\
EE(fixed)     & 0.940 & [ -0.207   0.994 ] \\
EE$(0.6)$     & 0.960 & [ 0.917   0.967 ] \\
\midrule
Spearman      & 0.935 & [ 0.899   0.960 ]\\
Pearson       & 0.999 \\
\bottomrule
\end{tabular}
\end{table}

\begin{figure}[t]
  \centering
  \includegraphics[width=\textwidth]{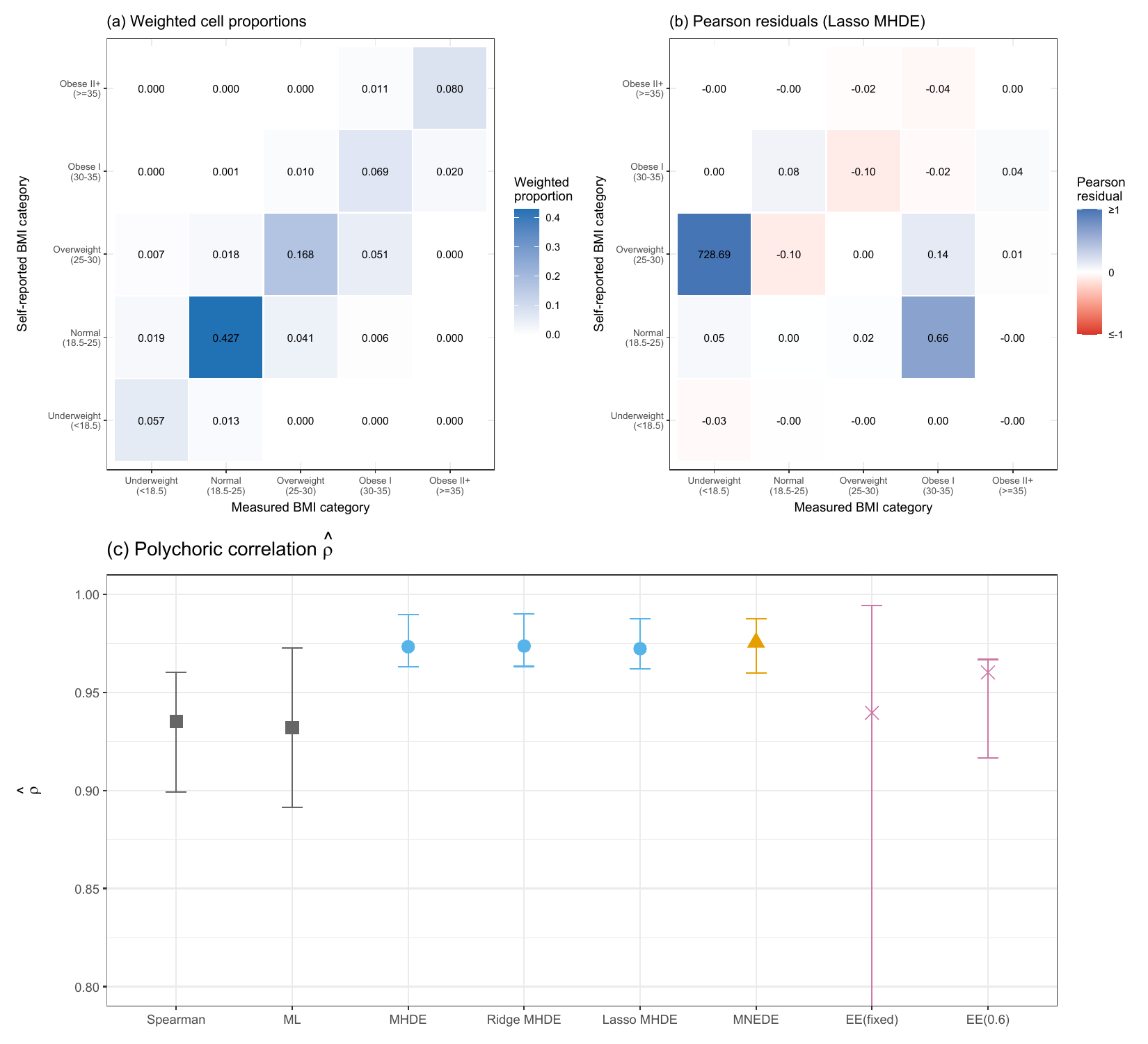}
  \caption{Self-reported vs. measured BMI, NHANES 2021--2023, ages 18--22.
    (a) design-weighted cell proportions; (b) Pearson
    residuals for the Lasso MHDE fit. (c) polychoric correlation by
    estimator with design-based 95\% confidence intervals; The large residual in the
    (self-overweight, measured-underweight) cell reflects a sparse,
    near-zero-probability cell that the robust criterion down-weights.}
  \label{fig:bmi}
\end{figure}

\section{Conclusion}

This paper develops robust estimators for polychoric correlation under complex survey designs based on two minimum divergence criteria: Hellinger distance with penalized Ridge and Lasso variants and the negative exponential disparity.
These estimators incorporate Horvitz--Thompson weighted cell frequencies to accommodate unequal inclusion probabilities arising from survey designs such as PPS sampling.
For MHDE, we established consistency and asymptotic normality with a design-dependent sandwich covariance. The influence function is finite on cells with positive probability but not uniformly bounded, reflecting Hellinger's sparse-cell sensitivity that drives the contrasting empirical performance across concordant and discordant contamination scenarios.

The empirical results under PPS sampling clearly show that the three classes of robust estimators considered here (MHDE, MNEDE, EE) offer complementary advantages: penalized MHDE methods provide the best estimation for the polychoric correlation under concordant contamination when marginals are correctly specified, while the E-estimator and MNEDE achieve lower MSE under discordant contamination or under compound misspecification--contamination effects.
These differences are primarily due to the different residual-adjustment mechanisms. MHDE has inverse-square-root sensitivity through $1/\sqrt{\pi_{ij}}$, whereas MNEDE uses a bounded negative-exponential adjustment of the Pearson residual and EE uses a sharp-cutoff rule.

Because these estimators respond differently to different contamination types, comparing them is informative in applied work where the contamination structure is unknown. Computing penalized MHDE alongside MNEDE (and EE) and comparing the resulting estimates provides a good sensitivity check: substantial disagreement---particularly between penalized MHDE and MNEDE or EE---is consistent with discordant contamination or a compound misalignment--contamination effect, and should prompt the practitioner to carefully inspect the weighted residual table for the cells driving the divergence and to report estimates from more than one estimator instead of committing to a single number.
Across the configurations examined, MNEDE remained competitive and is a sensible default when no contamination pattern is identified in advance, whereas penalized MHDE is preferable when the marginal parameters are themselves of inferential interest and standard Normal marginals of the latent variables are reasonable.
All methods are computationally tractable, providing practical tools for ordinal data analysis in complex survey designs.

Another advantage of the divergence-based estimators presented here is that they are tuning-free.
Even for the penalized variants, the regularization strength can be efficiently selected by cross-validation or other data-driven strategies.
For the EE, on the other hand, the tuning constant directly governs its robustness properties, making a data-driven choice for $c$ fragile and susceptible to contamination.
Moreover, because of the general nature of the proposed divergence-based estimators, they can be easily extended to more general polychoric models with arbitrary elliptical latent distributions \citep{Cheng2025} and higher dimensions.

\bibliographystyle{plainnat}
\bibliography{citation}

\appendix
\numberwithin{figure}{section}
\numberwithin{table}{section}
\numberwithin{equation}{section}

\section*{Appendix}
\addcontentsline{toc}{section}{Appendix}
\noindent The following appendices provide supporting material: parameterization
choices (Appendix~\ref{supp:param}), computational details
(Appendix~\ref{supp:computation}), the proof of Theorem~\ref{thm:hd_if}
(Appendix~\ref{supp:proofs}), marginal-parameter estimation
(Appendix~\ref{supp:nuisance}), survey-weighted inference
(Appendix~\ref{supp:weighted_ee}), and the compound non-standard-marginal
contamination effect (Appendix~\ref{supp:compound}).

\section{On the Choice of Parameterization}
\label{supp:param}

The bivariate polychoric model is not identifiable unless a normalization is imposed. In particular, without fixing either the latent marginal location-scale parameters or the threshold scale, there can exist two distinct parameter values that induce the same observed cell-probability vector:
\[
\boldsymbol{\psi}_1 \neq \boldsymbol{\psi}_2,
\qquad
\boldsymbol{\pi}(\boldsymbol{\psi}_1) = \boldsymbol{\pi}(\boldsymbol{\psi}_2).
\]
Identification therefore requires restricting the parameterization. Two natural normalizations are common: fixing the latent marginal distributions and estimating thresholds, or fixing the thresholds and estimating the latent marginal location-scale parameters. This section outlines the two formulations, identifies the empirical settings each is best suited to, and clarifies the relationship between them.

\subsection{Two parameterizations of the polychoric model}
\label{supp:two_param}

The classical formulation, adopted by \citet{Olsson1979} and the subsequent psychometric literature including \citet{WelzMairAlfons2026}, fixes the latent marginal distribution at standard normal,
\begin{equation*}
  (\xi, \eta)^{\top} \sim \mathcal{N}_2\!\left(
    \mathbf{0},\,
    \begin{pmatrix} 1 & \rho \\ \rho & 1 \end{pmatrix}
  \right),
\end{equation*}
and estimates the correlation coefficient $\rho$ jointly with the two sets of thresholds $\{a_k\}_{k=1}^{K-1}$ and $\{b_k\}_{k=1}^{K-1}$. The parameter vector under this formulation is $\boldsymbol{\theta}_{\mathrm{cl}} = (\rho, a_1, \ldots, a_{K-1}, b_1, \ldots, b_{K-1})^{\top}$, with dimension $2(K-1) + 1$.

The formulation adopted in this paper instead fixes the thresholds at values determined by the study design, and treats the latent marginal location and scale as free parameters,
\begin{equation*}
  (\xi, \eta)^{\top} \sim \mathcal{N}_2\!\left(
    \begin{pmatrix} \theta_1 \\ \theta_2 \end{pmatrix},\,
    \begin{pmatrix}
      \sigma_1^{2} & \rho \sigma_1 \sigma_2 \\
      \rho \sigma_1 \sigma_2 & \sigma_2^{2}
    \end{pmatrix}
  \right).
\end{equation*}
The parameter vector is $\boldsymbol{\theta}_{\mathrm{new}} = (\theta_1, \theta_2, \sigma_1, \sigma_2, \rho)^{\top}$, with fixed dimension irrespective of $K$. Both formulations achieve identification of the model and target the same population correlation coefficient between $\xi$ and $\eta$, since correlation is invariant under location and scale transformations.

\subsection{Application scenarios}
\label{supp:scenarios}

The two formulations correspond to distinct empirical contexts. The classical parameterization is appropriate when the ordinal response categories are subjective gradations of an unobservable continuous attribute, as in Likert-type psychometric items or attitude scales. In such settings, the discretization thresholds reflect how respondents map their internal experience onto fixed response options and have no external referent on a continuous scale; estimating them from the data while fixing the latent distribution at standard normal is the natural identification choice.

The parameterization adopted here is appropriate when the response categories are intervals on an externally measurable continuous variable, with cut-points determined by the survey instrument or by administrative criteria. The Current Population Survey \citep{CPS} collects household income in predetermined intervals defined in dollar amounts; the European Social Survey \citep{ESS} records education using the harmonized ISCED categories with externally specified boundaries. In both cases, the discretization thresholds are known design quantities, and treating them as estimable would discard information available \emph{a priori}. The marginal mean and variance, by contrast, describe genuine features of the target population that may depart substantially from the standard normal---right-skewed income distributions, cohort or country differences in mean education levels---and estimating them is both informative and necessary for the survey-based analyses considered in this paper.

\subsection{Identification and target parameter}
\label{supp:identification}

Although the two formulations differ in which parameters are held fixed, they target the same population correlation coefficient $\rho = \mathrm{Cor}(\xi, \eta)$. They differ in how residual variation from the data is absorbed during estimation: the classical formulation absorbs deviations of the empirical marginal distribution into the threshold estimates, while the formulation used here absorbs them into the marginal location and scale estimates. The empirical performance of $\hat{\theta}_1, \hat{\theta}_2, \hat{\sigma}_1, \hat{\sigma}_2$ in our framework is examined in Section~\ref{supp:nuisance}.

\subsection{Empirical comparison under the classical parameterization}
\label{supp:welz_empirical}

The discussion in Sections~\ref{supp:two_param}--\ref{supp:identification} characterizes the two parameterizations conceptually.
To complement that discussion empirically, we re-estimate the divergence-based estimators of the main text under the classical parameterization, fixing $\mu = 0$ and $\sigma = 1$ and treating the thresholds together with $\rho$ as free parameters, using the same simulation design described in Section~\ref{sec:simulations}.
Figure~\ref{fig:welz_param_comparison} shows the resulting MSE of $\hat{\rho}$ on a logarithmic scale across the six contamination scenarios. The figure also includes EE(fixed), a variant of the E-estimator that retains the bounded score function of \citet{WelzMairAlfons2026} but operates under the alternative parameterization, with thresholds fixed at the empirical marginal quantiles and the latent location and scale parameters treated as free.

In the concordant corner panels (upper and lower), the five classically-parameterized estimators behave similarly, with no
substantial differentiation among MHDE, MNEDE, and EE(0.6).
In the discordant corner, by contrast, MNEDE and EE(0.6) substantially outperform MHDE.
The exponential downweighting in MNEDE and the bounded score in EE(0.6) effectively suppress discordant-cell contamination under either parameterization, but the $1/\sqrt{\pi}$ scaling underlying the MHDE downweighting interacts poorly with threshold-only flexibility in the discordant regime.

EE(fixed) under the joint parameterization exhibits a fundamental instability: its MSE is above $0.3$ across all $\varepsilon$ and all scenarios, including the uncontaminated case.
This confirms that the joint parameterization is not suitable for the E-estimator's original formulation, and motivates the adoption of EE(0.6) under the classical parameterization in the main text's comparisons.

Taken together, the empirical comparison supports the parameterization choice made in the main text.
The joint parameterization preserves the benefits of penalization on the marginal parameters while maintaining the robustness behaviour of the HD-class estimators across both concordant and discordant contamination.
The classical parameterization, on the other hand, produces a regime in which the HD-class estimators lose their robustness advantage in the discordant corner.

\begin{figure}[htbp]
    \centering
    \includegraphics[width=\textwidth]{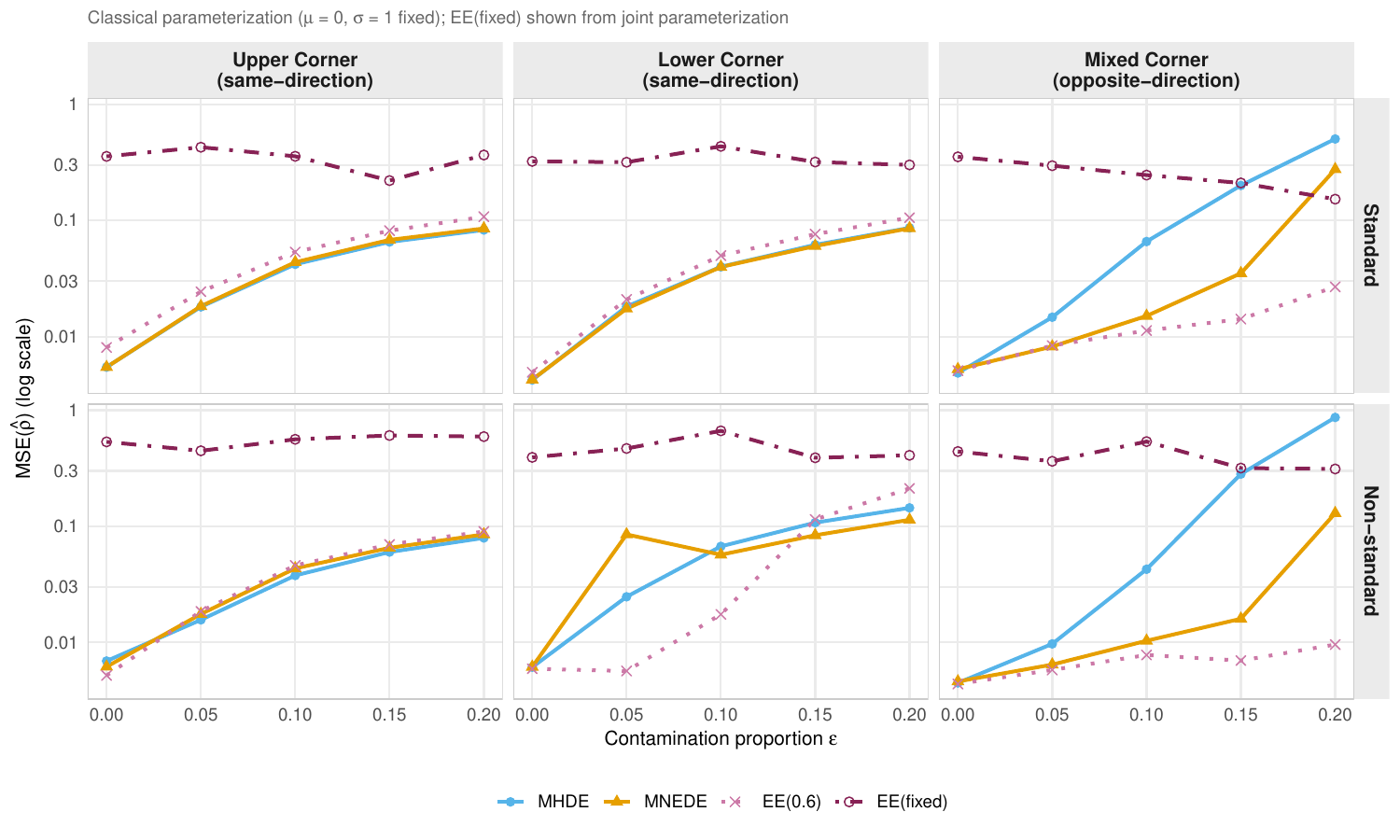}
    \caption{%
      MSE of $\hat{\rho}$ on a logarithmic scale across the six contamination scenarios.
      The estimators MHDE, MNEDE, and EE(0.6) are estimated under the classical parameterization, whereas EE(fixed) is a variant of the E-estimator constructed by combining the bounded score function of \citet{WelzMairAlfons2026} with the alternative parameterization used in the main text.}
    \label{fig:welz_param_comparison}
\end{figure}

\section{Computational Details}
\label{supp:computation}

\subsection{Optimization algorithm}

We minimize the Hellinger distance $D_\text{HD}(f_{\boldsymbol{\psi}}, \hat{g}^w) = \frac{1}{2}\sum_{i,j} (\sqrt{\hat{p}_{ij}^w} - \sqrt{\pi_{ij}(\boldsymbol{\psi})})^2$ using the BFGS quasi-Newton algorithm via R's \texttt{optim()} function, with convergence tolerance $10^{-8}$, maximum 1000 iterations, and numerical stability constant $\epsilon = 10^{-10}$ added to theoretical probabilities before taking square roots.

Starting values are obtained through a two-stage procedure. First, we attempt ML estimation on an outlier-cleaned subset: compute weighted cell frequencies $\hat{p}_{ij}^w$, identify non-sparse cells with $\hat{p}_{ij}^w > 0.05$, and fit the ML estimator using only those cells. The resulting estimates serve as starting values for the full HD optimization. If ML initialization fails, we use default starting values $\theta_i^{(0)} = 0$, $\sigma_i^{(0)} = 1$, and $\rho^{(0)} = 0.5$.

\subsection{Parameter transformation}

During optimization, we transform parameters to ensure they remain in their valid ranges:
\begin{align*}
\theta_i^{(\text{opt})} &= \theta_i \quad (\text{unbounded}) \\
\sigma_i^{(\text{opt})} &= \log(\sigma_i) \quad (\text{ensures } \sigma_i > 0) \\
\rho^{(\text{opt})} &= \text{atanh}(\rho) = \tfrac{1}{2}\log\!\left(\tfrac{1+\rho}{1-\rho}\right) \quad (\text{ensures } |\rho| < 1).
\end{align*}
The inverse transformations are $\theta_i = \theta_i^{(\text{opt})}$, $\sigma_i = \exp(\sigma_i^{(\text{opt})})$, and $\rho = \tanh(\rho^{(\text{opt})})$.

\subsection{Convergence diagnostics}

We verify convergence by checking that the gradient norm satisfies $\|\nabla D_\text{HD}(\hat{\boldsymbol{\psi}})\| < 10^{-6}$, the parameters lie within valid ranges ($\sigma_i \in (0.1, 10)$, $|\rho| < 0.99$), and the Hellinger distance is non-degenerate ($D_\text{HD}(\hat{\boldsymbol{\psi}}) < 0.99$). Replications failing any convergence check are excluded from Monte Carlo summaries.

\subsection{Computational complexity}

The per-iteration cost is $O(K^2)$ for ML, MHDE, MNEDE, and EE (dominated by bivariate normal CDF evaluations for cell probabilities). EE estimation is implemented via the \texttt{robcat} R package \citep{robcat2026}, with the HT-weighted contingency table constructed via \texttt{survey::svytable} \citep{Lumley2024} as input. Penalized variants additionally require an $O(L \times F)$ factor for the lambda grid ($L = 15$ values) and cross-validation folds ($F = 5$). With the survey data generation step, each Monte Carlo replication takes approximately 5--10 seconds per estimator for $K = 5$, $N = 5{,}000$, and $n \approx 500$.

\section{Proof of Theorem~\ref{thm:hd_if}: Influence function of the MHDE functional}
\label{supp:proofs}

\begin{proof}
Let
\[
\mathbf U(\boldsymbol{\psi}, \mathbf p) 
= \nabla_{\boldsymbol{\psi}} D_{\mathrm{HD}}(\boldsymbol{\pi}(\boldsymbol{\psi}), \mathbf p)
\]
denote the population Hellinger estimating equation. By definition of $\boldsymbol{\psi}_g = T(G)$,
$\mathbf U(\boldsymbol{\psi}_g, \mathbf g) = \mathbf 0$.
For point-mass contamination at the cell $y = (i, j)$, write
\[
\mathbf g_{\varepsilon} 
= (1 - \varepsilon)\mathbf g + \varepsilon \mathbf e_{ij}
= \mathbf g + \varepsilon(\mathbf e_{ij} - \mathbf g),
\]
and let $\boldsymbol{\psi}_{\varepsilon} = T(G_{\varepsilon})$ be the corresponding population minimizer.
Then $\mathbf U(\boldsymbol{\psi}_{\varepsilon}, \mathbf g_{\varepsilon}) = \mathbf 0$.
Differentiating this identity with respect to $\varepsilon$ and evaluating at $\varepsilon = 0$ gives
\[
\nabla_{\boldsymbol{\psi}} \mathbf U(\boldsymbol{\psi}, \mathbf g)\Big|_{\boldsymbol{\psi} = \boldsymbol{\psi}_g}
\frac{\mathrm d\, \boldsymbol{\psi}_{\varepsilon}}{\mathrm d\,\varepsilon}\Big|_{\varepsilon = 0}
+ \nabla_{\mathbf p} \mathbf U(\boldsymbol{\psi}_g, \mathbf p)\Big|_{\mathbf p = \mathbf g}
(\mathbf e_{ij} - \mathbf g) = \mathbf 0.
\]
By the definition of $\mathbf A$ and $\boldsymbol{\phi}(i, j)$, this becomes $\mathbf A\,\mathrm{IF}(i, j; T, G) + \boldsymbol{\phi}(i, j) = \mathbf 0$.
Since $\mathbf A$ is nonsingular,
\[
\mathrm{IF}(i, j; T, G) = -\mathbf A^{-1}\boldsymbol{\phi}(i, j).
\]

It remains only to interpret the finiteness claim.
Since $K$ is fixed, the support $\mathcal Y = \{1, \ldots, K\}^{2}$ is finite.
If $g_{ij} > 0$ for all cells and the fitted probabilities $\pi_{ij}(\boldsymbol{\psi}_g)$ are also positive, then the derivatives in $\boldsymbol{\phi}(i, j)$ are finite for every cell.
Hence $\|\mathrm{IF}(i, j; T, G)\| < \infty$ for each fixed cell $(i, j)$.

This finiteness is a pointwise finite-support statement.
It should not be read as uniform boundedness over sparse-cell regimes.
Indeed, the derivative of the Hellinger estimating equation contains inverse square-root factors involving $g_{ij}$ and $\pi_{ij}(\boldsymbol{\psi}_g)$, and these terms can become large when either the population or model-implied cell probability is close to zero.
Thus the influence function is finite for any fixed table with strictly positive cell probabilities, but it need not be uniformly bounded as some cell probabilities approach zero.
\end{proof}

\section{Estimation of Marginal Parameters $\theta_1$ and $\sigma_1$}
\label{supp:nuisance}

The main paper focuses on $\hat{\rho}$, the bivariate parameter of primary interest. Here we examine the corresponding behaviour for the univariate marginal parameters $\theta_1$ and $\sigma_1$, both as MSE summaries (Section~\ref{supp:mse_univ}) and as full sampling distributions across replications (Sections~\ref{supp:box_theta1}--\ref{supp:box_sigma1}). Section~\ref{supp:bv_tradeoff} synthesizes the bias--variance trade-off implied by these results.

\subsection{MSE summaries}
\label{supp:mse_univ}

Figures~\ref{fig:mse_theta1_panel} and~\ref{fig:mse_sigma1_panel} present MSE results for $\theta_1$ and $\sigma_1$ across all six scenarios. EE(0.6) is omitted from these figures because the formulation of \citet{WelzMairAlfons2026} assumes standardized latent variables ($\theta_l = 0$, $\sigma_l = 1$) by design and therefore does not estimate $(\theta_1, \sigma_1)$.

Under standard marginals with concordant contamination (top row, left two columns), penalized MHDE methods---particularly Ridge MHDE and Lasso MHDE---achieve substantially lower MSE for both $\theta_1$ and $\sigma_1$ compared to ML, MNEDE, and unpenalized MHDE. At $\varepsilon = 0.20$, Lasso MHDE achieves MSE$(\theta_1) \approx 0.04$ and MSE$(\sigma_1) \approx 0.02$ compared to ML at $\approx 0.25$ and $\approx 0.13$ respectively. Under mixed contamination, MNEDE and penalized MHDE variants all maintain low MSE, while ML and unpenalized MHDE degrade.

Under non-standard marginal (bottom row), the advantage of penalized MHDE diminishes for $\theta_1$ as penalization toward the wrong target ($\theta = 0$ when truth is $\theta = 0.5$) becomes counterproductive. For $\sigma_1$, penalized MHDE retains some advantage under concordant contamination but loses effectiveness under the compound non-standard marginal contamination effect in the lower corner scenario.

\begin{figure}[htbp]
    \centering
    \includegraphics[width=\textwidth]{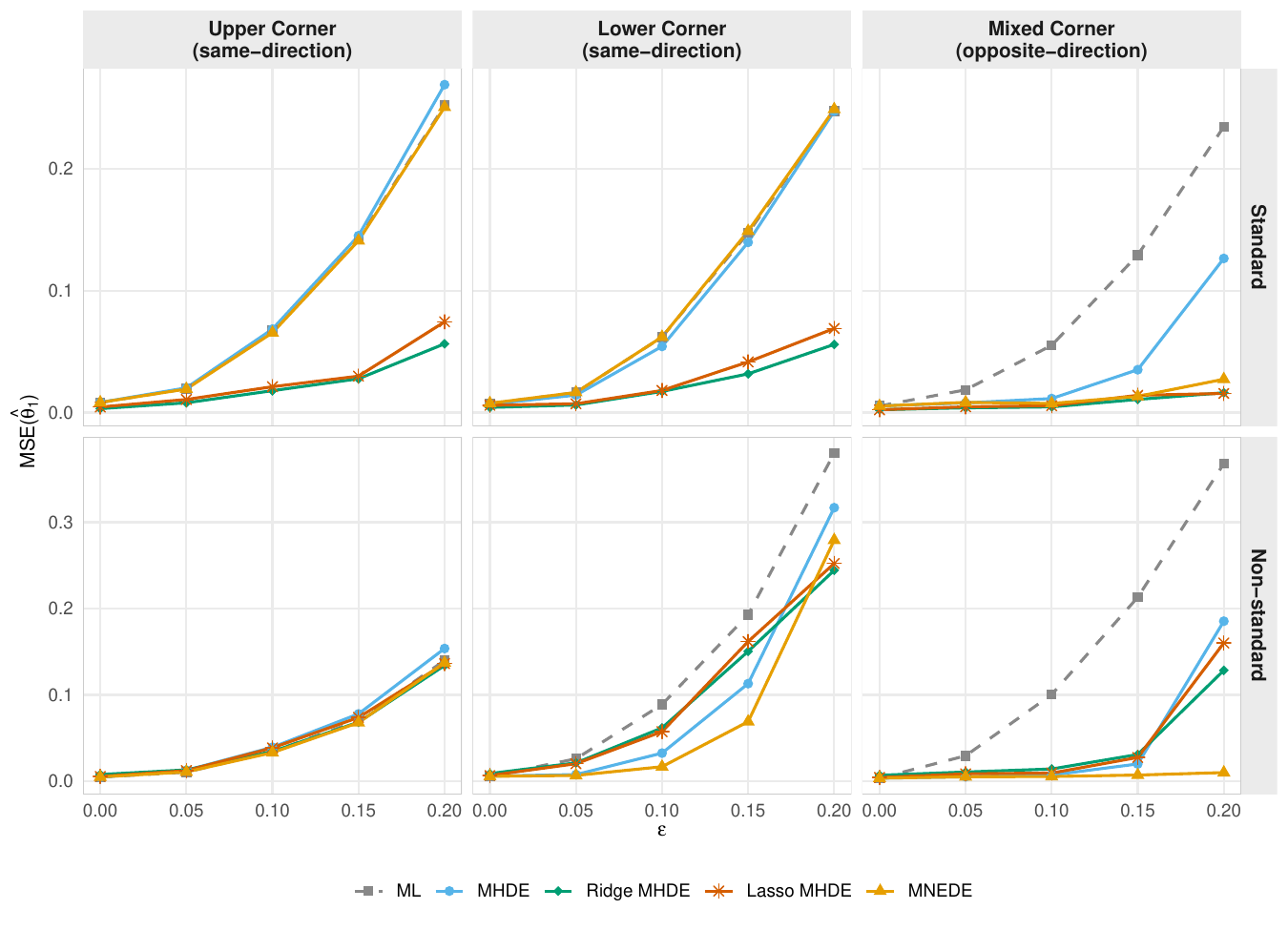}
    \caption{MSE of $\hat{\theta}_1$ across all six scenarios. EE(0.6) is omitted because the \citet{WelzMairAlfons2026} formulation assumes standardized latent variables and does not estimate $\theta_1$. Penalized MHDE methods dominate under standard concordant scenarios; under non-standard marginal, their advantage diminishes as penalization targets the wrong standardization point.}
    \label{fig:mse_theta1_panel}
\end{figure}

\begin{figure}[htbp]
    \centering
    \includegraphics[width=\textwidth]{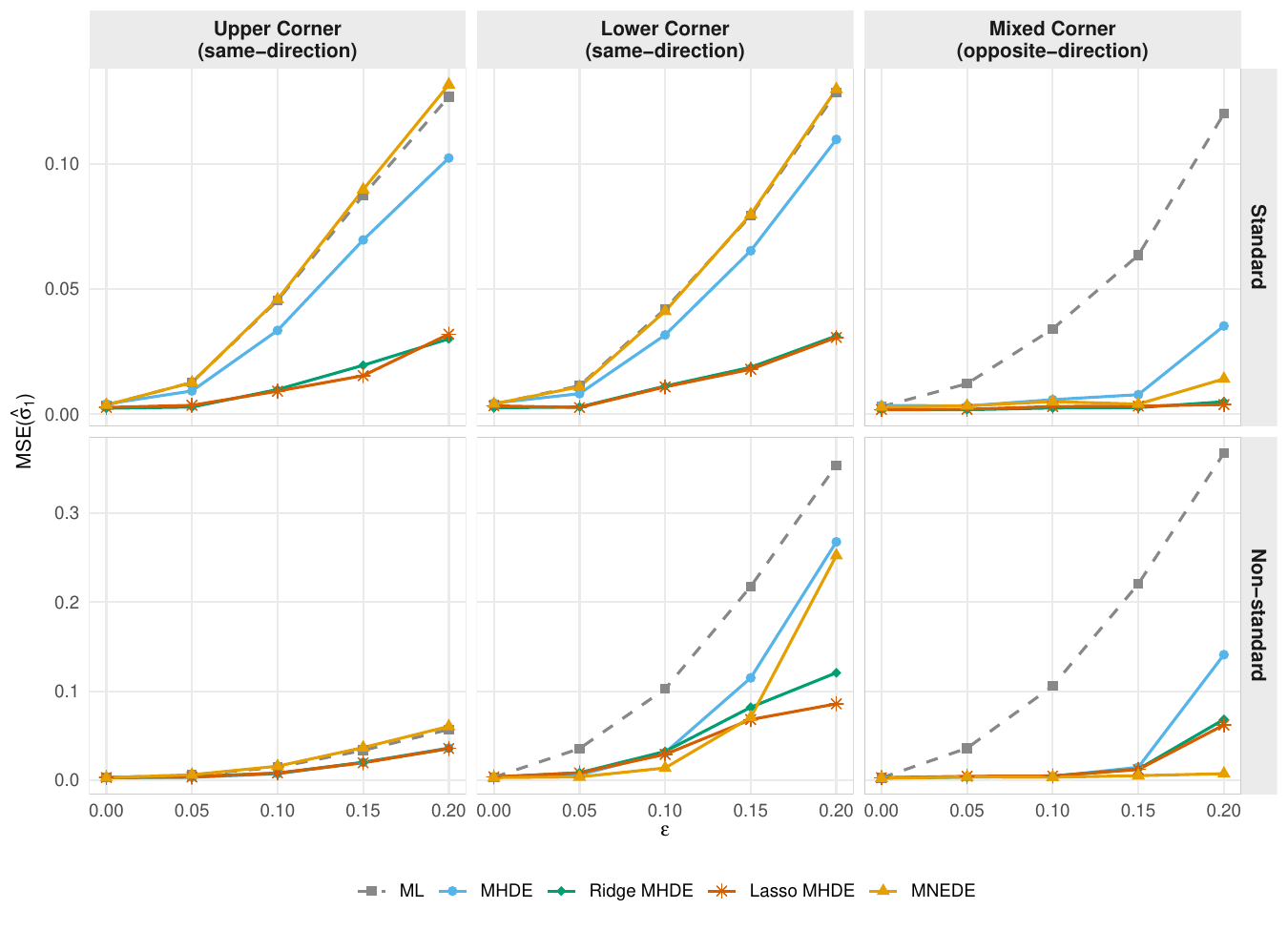}
    \caption{MSE of $\hat{\sigma}_1$ across all six scenarios. EE(0.6) is omitted because the \citet{WelzMairAlfons2026} formulation assumes standardized latent variables and does not estimate $\sigma_1$. Under standard marginals, Lasso MHDE achieves MSE reductions of 80\% or more relative to ML at $\varepsilon = 0.20$ for concordant contamination.}
    \label{fig:mse_sigma1_panel}
\end{figure}

\subsection{Distribution of $\hat{\theta}_1$ across replications}
\label{supp:box_theta1}

MSE summaries collapse the full sampling distribution into bias and variance moments, masking distributional features such as outliers, skewness, and tail behaviour that are relevant when applying an estimator to a single dataset. Figure~\ref{fig:box_theta1} displays the distribution of $\hat{\theta}_1 - \theta_1$ across the 100 Monte Carlo replications at each contamination level for the five HD-, NED-, and ML-based estimators (EE omitted; see Section~\ref{supp:bv_tradeoff}).

Under standard marginals (top row), all estimators are centred near zero at $\varepsilon = 0$, with tight interquartile ranges. As contamination increases, ML and unpenalized MHDE develop systematic bias whose direction reflects the contaminated corner: positive bias under upper corner contamination and negative bias under lower corner contamination, with absolute median bias reaching $0.4$ at $\varepsilon = 0.20$. Penalized MHDE methods (Ridge MHDE and Lasso MHDE) and MNEDE maintain medians close to zero across all $\varepsilon$ values in these concordant scenarios. Under mixed contamination (top right), the bias trends largely cancel and all five estimators maintain medians close to zero, but MNEDE's interquartile range visibly widens with $\varepsilon$ while the HD-based methods stay tight.

Under the non-standard marginal configuration (bottom row, true $\theta_1 = 0.5$), all joint-parameterization estimators recover $\theta_1$ with negligible bias at $\varepsilon = 0$ confirming that estimating $\theta$ and $\sigma$ directly leaves the model correctly specified under shifted and compressed latent marginals. Contamination then perturbs this baseline differently by corner: upper-corner contamination drives $\hat{\theta}_1$ upward (less negative)  for ML and MHDE, lower-corner contamination drives it downward ($\approx -0.55$ to $-0.60$), and mixed-corner contamination, which places contaminated mass in the low category of $Y_1$, likewise produces a downward drift. As for $\hat{\rho}$, MNEDE again exhibits the widest interquartile range among the robust methods, most visibly at intermediate $\varepsilon$.

\begin{figure}[htbp]
    \centering
    \includegraphics[width=\textwidth]{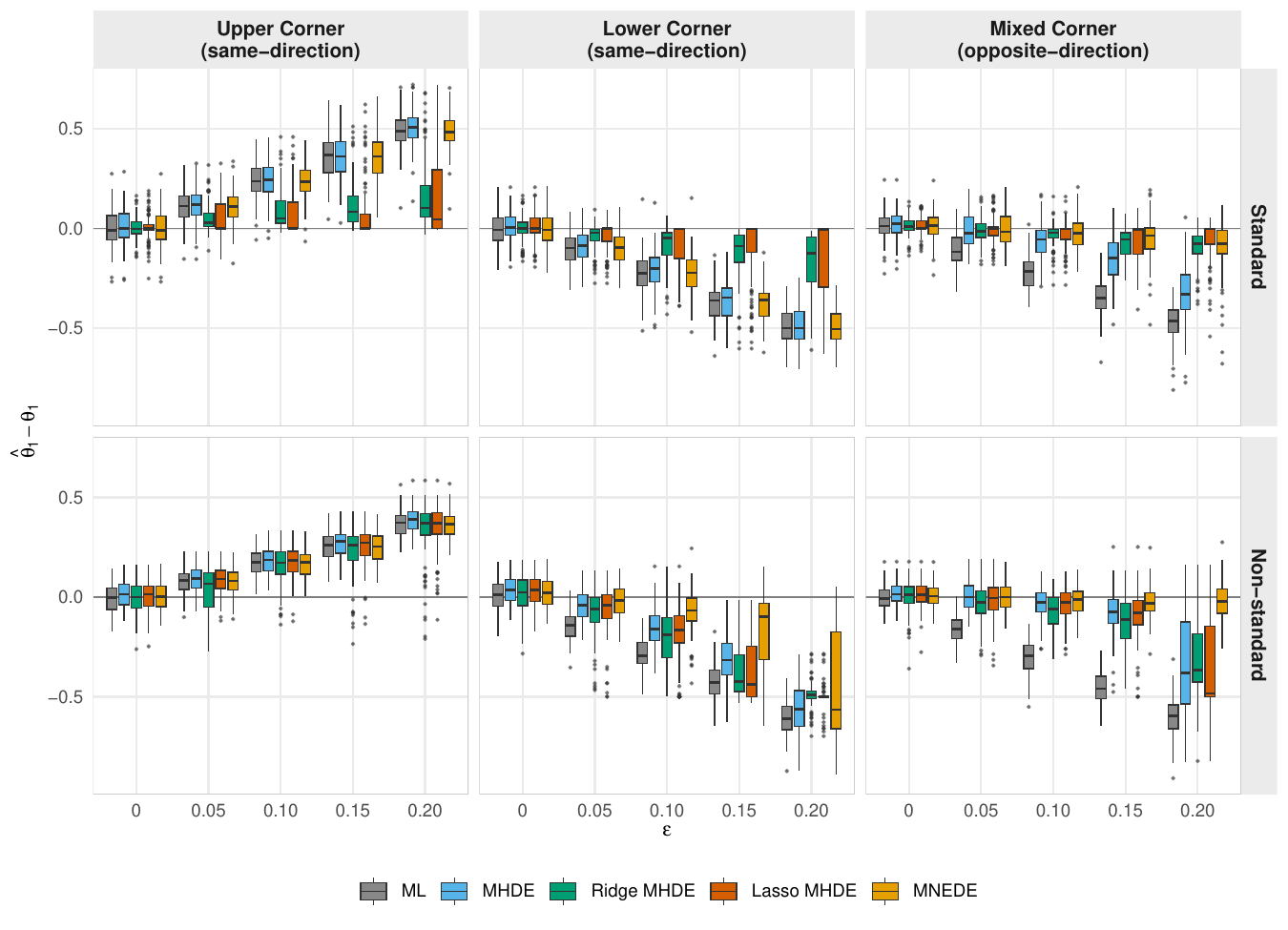}
    \caption{Distribution of $\hat{\theta}_1 - \theta_1$ across 100 Monte Carlo replications for each $\varepsilon$ level, faceted by scenario. Boxes show the interquartile range; whiskers extend to $1.5 \times$ IQR; points beyond are individual replications. EE(0.6) is omitted because the \citet{WelzMairAlfons2026} formulation does not estimate $\theta_1$.}
    \label{fig:box_theta1}
\end{figure}

\subsection{Distribution of $\hat{\sigma}_1$ across replications}
\label{supp:box_sigma1}

Figure~\ref{fig:box_sigma1} shows the corresponding distributions for $\hat{\sigma}_1 - \sigma_1$. The bias is uniformly positive across all estimators and all contaminated settings: corner contamination, regardless of which corner is affected, expands the empirical support of the contingency table relative to the model's expected support, and all estimators compensate by inflating $\hat{\sigma}_1$. The magnitude of positive bias scales with $\varepsilon$, and the ordering by estimator is consistent across scenarios---ML produces the largest bias, MNEDE is intermediate with notably wider interquartile range, and penalized MHDE methods (Ridge MHDE and Lasso MHDE) yield the most compact distributions with the smallest median bias. The pattern holds under both standard and non-standard marginals, although the absolute magnitudes differ.

\begin{figure}[htbp]
    \centering
    \includegraphics[width=\textwidth]{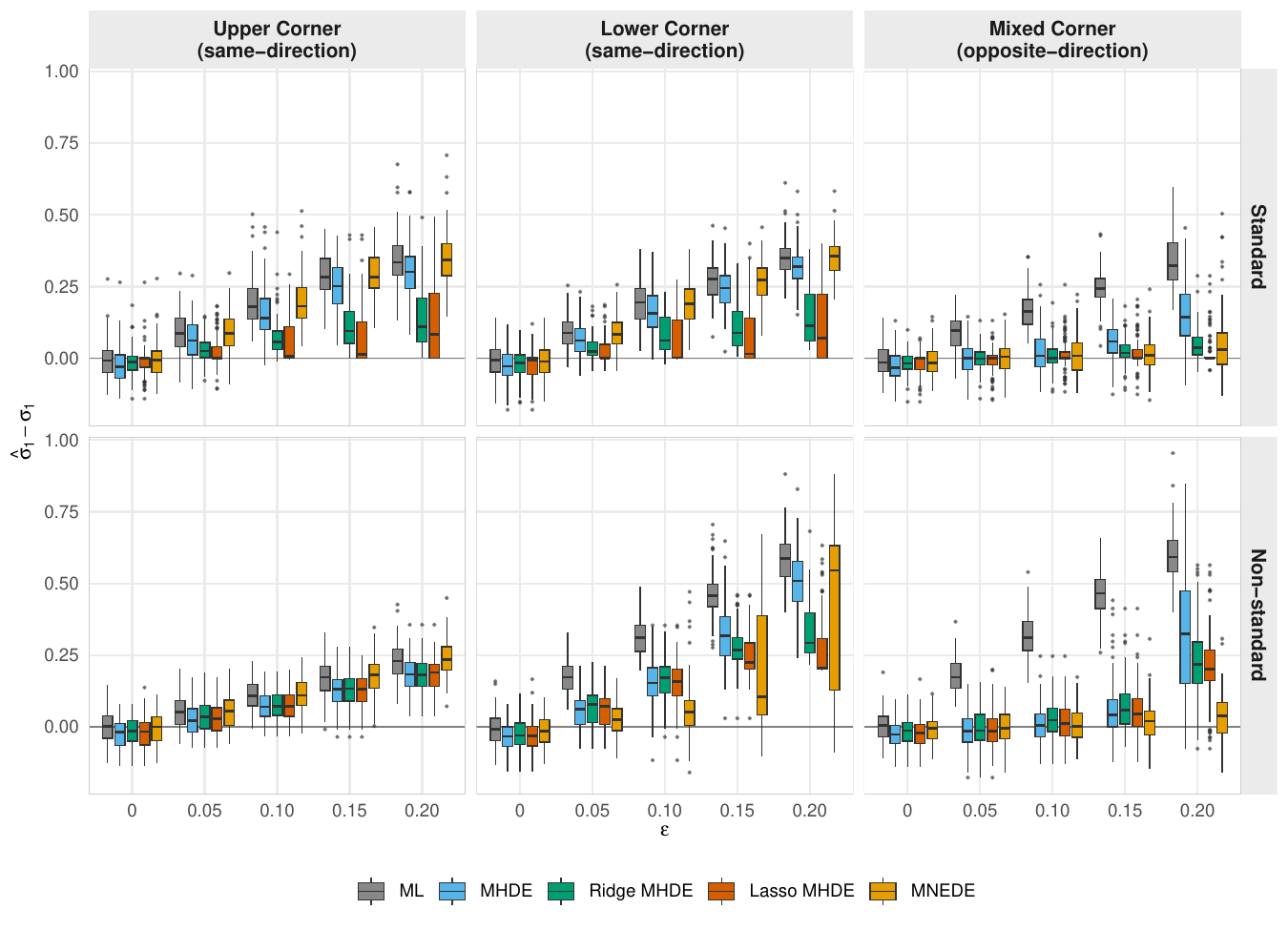}
    \caption{Distribution of $\hat{\sigma}_1 - \sigma_1$ across 100 Monte Carlo replications for each $\varepsilon$ level, faceted by scenario. Boxes show the interquartile range; whiskers extend to $1.5 \times$ IQR; points beyond are individual replications. EE(0.6) is omitted because the \citet{WelzMairAlfons2026} formulation does not estimate $\sigma_1$.}
    \label{fig:box_sigma1}
\end{figure}

\subsection{Bias--variance trade-off}
\label{supp:bv_tradeoff}

The boxplots in Figures~\ref{fig:box_theta1} and~\ref{fig:box_sigma1} reveal a bias--variance trade-off for the marginal parameters that complements, rather than contradicts, the recommendations made in the main text for $\hat{\rho}$. Three observations support this characterization. First, penalized MHDE methods maintain the tightest interquartile ranges across all scenarios, reflecting the variance-stabilizing effect of nuisance-parameter shrinkage. Second, MNEDE, which attains competitive or better MSE for $\hat{\rho}$ under mixed contamination and the compound non-standard marginal contamination effect, exhibits noticeably wider spread for both $\hat{\theta}_1$ and $\hat{\sigma}_1$ than the MHDE variants, most visibly at intermediate contamination levels ($\varepsilon = 0.10$--$0.15$). Third, the magnitude of the trade-off scales with contamination severity: at $\varepsilon = 0$ all estimators have comparable spread, while by $\varepsilon = 0.20$ the IQR of MNEDE can exceed that of Lasso MHDE by a factor of two or more for the marginal parameters.

This pattern is consistent with MNEDE's residual-adjustment mechanism.
The negative-exponential RAF bounds the contribution of cells with large positive Pearson residuals, thereby reducing the effect of cells whose HT-weighted empirical frequency substantially exceeds the fitted model probability.
This protection can also reduce the contribution of legitimate sampling fluctuations in cells with large residuals, increasing variability of the marginal parameter estimates in some scenarios.
In practical terms, the choice between Lasso MHDE and MNEDE for $\hat{\rho}$ does not directly translate to the marginal parameters: under conditions where MNEDE achieves competitive or better MSE for $\hat{\rho}$, Lasso MHDE continues to provide more stable estimates of $(\theta_1, \sigma_1)$.
Researchers whose inferential targets include the marginal parameters---or whose downstream applications rely on the joint distribution rather than $\rho$ alone---should weigh this trade-off when selecting an estimator.

\section{Survey-Weighted Inference and Comparison with i.i.d.\ Weighting}
\label{supp:weighted_ee}

\subsection{Design-based versus i.i.d.\ variance formulas}
\label{supp:variance_compare}
The asymptotic covariance derived in Theorem~\ref{thm:hd_asymptotic} takes the sandwich form
\[
\mathbf A^{-1}\mathbf B_d\mathbf A^{-\top},
\qquad
\mathbf B_d=\mathbf J_g\mathbf V_d\mathbf J_g^\top,
\]
where \(\mathbf V_d\) is the limiting design-adjusted covariance matrix of the ratio-normalized HT cell-frequency vector.
Thus the sampling design enters through \(\mathbf V_d\).

An approach that substitutes the weighted empirical frequencies $\hat{p}_{ij}^w$ directly into a loss function derived under i.i.d.\ sampling and then applies the i.i.d.\ variance formula instead yields a covariance that ignores the design contribution.

To compare the two formulas concretely, write the score under HT weighting as $\boldsymbol{U}_n = n^{-1/2}\sum_{k \in S} w_k \boldsymbol{\phi}(Y_k)$. Under Poisson-PPS sampling, the design-based variance of this score is
\[
\mathrm{Var}_{\mathrm{design}}(\boldsymbol{U}_n) = \frac{1}{n}\sum_{k \in U}\frac{1 - P(k \in S_\gamma)}{P(k \in S_\gamma)}\,\boldsymbol{\phi}(Y_k)\boldsymbol{\phi}(Y_k)^\top,
\]
where $P(k \in S_\gamma)$ is the inclusion probability of unit $k$.
The i.i.d.\ variance formula treats the weighted frequencies as if drawn from a multinomial under a single distribution $g$:
\[
\mathrm{Var}_{\mathrm{iid}}(\boldsymbol{U}_n) = \mathbb{E}_g\!\left[\boldsymbol{\phi}(Y)\boldsymbol{\phi}(Y)^\top\right].
\]
The two formulas agree only when all $\pi_k$ are equal (simple random sampling); otherwise, $\mathrm{Var}_{\mathrm{design}}$ exceeds $\mathrm{Var}_{\mathrm{iid}}$ whenever the weights $w_k = 1/P(k \in S_\gamma)$ have positive variance.

The discrepancy is quantified by the \emph{design effect}
\[
\mathrm{DEFF} = \frac{\mathrm{Var}_{\mathrm{design}}(\hat{\boldsymbol{\psi}})}{\mathrm{Var}_{\mathrm{iid}}(\hat{\boldsymbol{\psi}})},
\]
which to first order satisfies $\mathrm{DEFF} \approx 1 + \mathrm{CV}^2(w)$ for estimators of population means \citep{Kish1965}. The effective sample size $n_\mathrm{eff} = (\sum_k w_k)^2 / \sum_k w_k^2$ captures the same quantity through the equivalence $\mathrm{DEFF} = n / n_\mathrm{eff}$.

In the simulation design considered in this paper, the weight coefficient of variation is $\mathrm{CV}(w) \approx 1.26$ (Section~\ref{supp:design}), corresponding to $\mathrm{DEFF} \approx 2.5$. Standard errors computed under the i.i.d.\ formula would therefore be approximately $\sqrt{2.5} \approx 1.58$ times smaller than the correct design-based standard errors, and confidence intervals constructed from them would undercover the nominal level by a corresponding margin. The design-aware variance in Theorem~\ref{thm:hd_asymptotic} eliminates this discrepancy through the explicit dependence on the inclusion probabilities.

\subsection{Survey-weighted versus unweighted E-estimation: point estimates}
\label{supp:weighted_ee_points}

A separate question is whether the point estimates of EE differ when the contingency table is supplied with or without survey weights. The E-estimator of \citet{WelzMairAlfons2026} is formulated for an unweighted contingency table from i.i.d.\ ordinal observations. To accommodate complex survey designs, we instead supply the HT-weighted contingency table constructed via \texttt{survey::svytable} \citep{Lumley2024}, scaled and rounded to integers. Because the E-estimator's loss depends only on the relative empirical frequencies $\hat{p}_{ij}^w / \pi_{ij}(\boldsymbol{\psi})$, this preserves design-weighted point estimation; but standard errors reported by \texttt{robcat} do not account for survey weights and are therefore not used.

For comparison, we also applied the unweighted variant of EE---supplying the raw contingency table $\sum_{k \in S} \mathbf{1}(Y_{1k} = i, Y_{2k} = j)$ without weights---and Table~\ref{tab:ee_weighted_vs_unweighted} reports the MSE for $\hat{\rho}$ at $\varepsilon = 0.20$ across all six scenarios.

\begin{table}[htbp]
\centering
\caption{MSE of $\hat{\rho}$ at $\varepsilon = 0.20$ for survey-weighted versus unweighted E-estimation across the six simulation scenarios.}
\label{tab:ee_weighted_vs_unweighted}
\small
\begin{tabular}{llccc}
\toprule
\textbf{Marginals} & \textbf{Cont.} & \textbf{Weighted EE(0.6)} & \textbf{Unweighted EE(0.6)} & \textbf{Difference} \\
\midrule
Standard & Upper & 0.110 & 0.079 & $+0.031$ \\
Standard & Lower & 0.103 & 0.110 & $-0.008$ \\
Standard & Mixed & 0.024 & 0.013 & $+0.010$ \\
\midrule
Non-standard & Upper & 0.092 & 0.059 & $+0.034$ \\
Non-standard & Lower & 0.089 & 0.046 & $+0.042$ \\
Non-standard & Mixed & 0.010 & 0.004 & $+0.006$ \\
\bottomrule
\end{tabular}
\end{table}

Two patterns emerge, with the gap between weighted and unweighted EE driven by different MSE components in different contamination regimes.
Under concordant contamination, the survey-weighted EE has higher MSE than its unweighted counterpart, and decomposing MSE into squared bias and variance reveals that the increase is driven primarily by bias.
The mechanism is straightforward: PPS weights $w_i = 1/P(i \in S)$ are determined by the auxiliary variable $X_i$ used in the inclusion probability formula, independent of the contamination mechanism. When a high-weight unit is contaminated and reassigned to the concordant corner cell $(K, K)$ or $(1, 1)$, that unit contributes proportionally more weighted mass to the contaminated cell than its unweighted counterpart, amplifying the bias of any estimator that does not engage robust downweighting in this regime---including EE under concordant contamination, where Pearson residuals stay close to one and the E-estimator's hard-cutoff is not triggered.

Under mixed-corner contamination, the survey-weighted EE also has marginally higher MSE than its unweighted counterpart, but for a different reason. Both variants achieve very low bias because the E-estimator's downweighting mechanism engages fully in this regime, so MSE is dominated by sampling variance rather than bias. The increase is driven by the standard design effect of unequal PPS weights (with weight coefficient of variation around $1.26$ in our simulations), which inflates the variance of the weighted estimator relative to the unweighted variant. This is a well-known precision cost of incorporating design weights and is independent of the contamination geometry.

We adopt the survey-weighted variant as our primary results because it is consistent with the design-aware treatment of the other estimators in this paper. The unweighted variant is reported here for transparency and to illustrate that point estimation under complex survey designs is sensitive to whether weights are incorporated into the contingency table input.

\subsection{Survey Design Diagnostics}
\label{supp:design}

We report two diagnostics for the Poisson-PPS sampling design: the achieved sample size $n$ and the design efficiency, measured by the effective sample size $n_\text{eff} = (\sum_i w_i)^2 / \sum_i w_i^2$ and the coefficient of variation of the weights $\text{CV}(w) = \text{sd}(w_i)/\bar{w}$.

Table~\ref{tab:design_diagnostics} and Figure~\ref{fig:design_diag} summarize these quantities across contamination levels. The effective sample size $n_\text{eff}$ captures the variance inflation due to unequal weights: a ratio $n_\text{eff}/n < 1$ indicates precision loss relative to a simple random sample. The weight CV quantifies the degree of unequal weighting; larger values correspond to more variable inclusion probabilities, which amplify the influence of contaminated units that happen to carry high weights. Contamination is applied after sampling and does not affect the sampling design itself; results are therefore averaged over all six scenarios.

\begin{table}[htbp]
\centering
\caption{Summary of survey design diagnostics under Poisson-PPS sampling ($N = 5{,}000$, $n \approx 500$, $\rho_{YZ} = 0.25$). Values are means across 100 replications per $\varepsilon$ level, averaged over all six scenarios.}
\label{tab:design_diagnostics}
\begin{tabular}{ccccc}
\toprule
$\varepsilon$ & Mean $n$ & Mean $n_\text{eff}$ & $n_\text{eff}/n$ & Mean CV$(w)$ \\
\midrule
0.00 & 496 & 194 & 0.391 & 1.26 \\
0.05 & 495 & 194 & 0.392 & 1.26 \\
0.10 & 497 & 195 & 0.392 & 1.26 \\
0.15 & 493 & 193 & 0.392 & 1.27 \\
0.20 & 497 & 195 & 0.393 & 1.27 \\
\bottomrule
\end{tabular}
\end{table}

\begin{figure}[htbp]
    \centering
    \includegraphics[width=\textwidth]{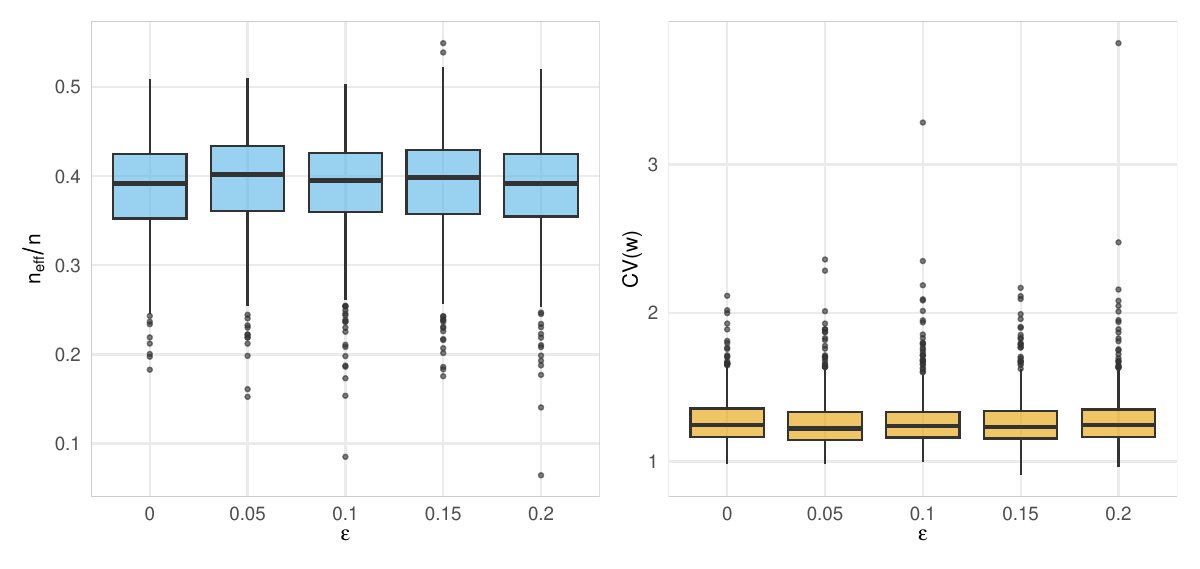}
    \caption{Survey design diagnostics. Left panel: distribution of effective sample size ratio $n_\text{eff}/n$ by $\varepsilon$. Right panel: distribution of weight coefficient of variation CV$(w)$ by $\varepsilon$.}
    \label{fig:design_diag}
\end{figure}

\section{Compound Non-standard Marginal Contamination Effect on $\hat{\rho}$}
\label{supp:compound}

Figure~\ref{fig:misspec_lower_detail} provides a detailed view of $\hat{\rho}$ under the non-standard marginals with lower corner contamination, where EE(0.6) and MNEDE both substantially outperform the HD-based variants. EE(0.6) achieves the lowest MSE across all $\varepsilon > 0$ with MNEDE close behind. The left panel shows that EE(0.6) and MNEDE both maintain substantially lower bias than the HD-based methods across all $\varepsilon$; the right panel shows that the gap with HD-based variants persists throughout the contamination range. The compound effect arises because non-standard marginal makes the assumed model's $\pi_{11}$ smaller than the true value, amplifying HD's $1/\sqrt{\pi}$ sensitivity to contamination in this cell, while MNEDE's exponential suppression and EE's sharp-cutoff downweighting both bypass the $1/\sqrt{\pi}$ scaling and remain effective regardless of non-standard marginal induced sparsity.

\begin{figure}[htbp]
    \centering
    \includegraphics[width=\textwidth]{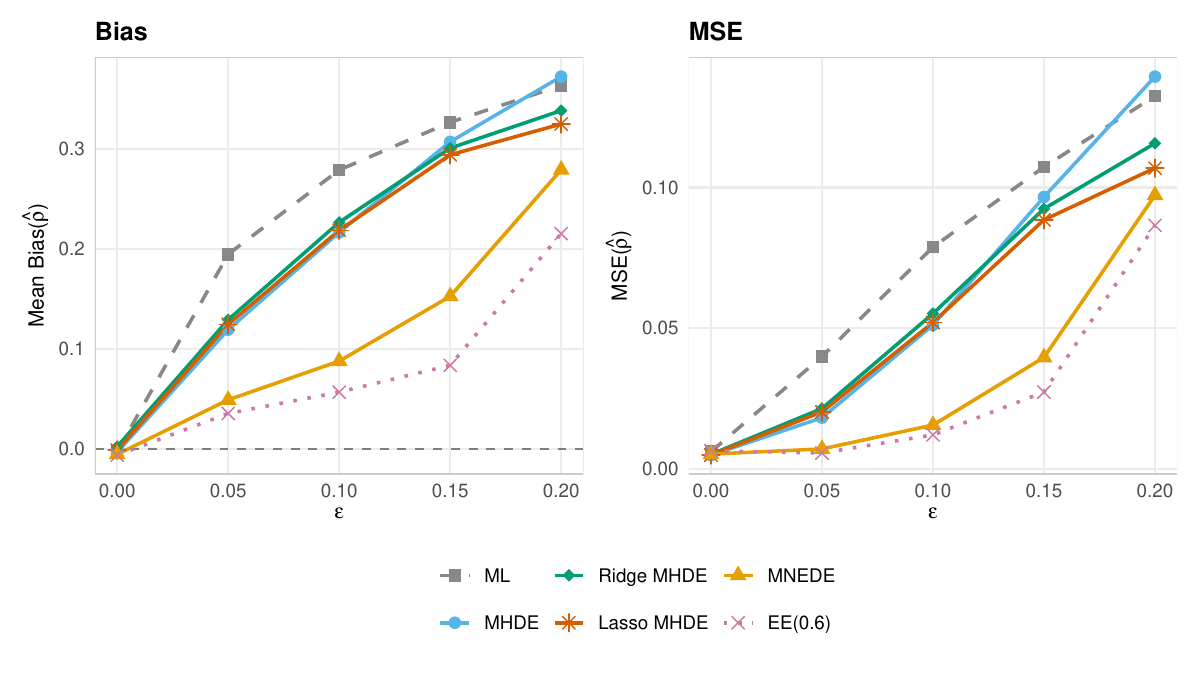}
    \caption{Bias (left) and MSE (right) of $\hat{\rho}$ under non-standard marginals with lower corner contamination. EE(0.6) achieves the lowest MSE across all $\varepsilon$ values, with MNEDE second-best, both outperforming the HD-based variants substantially due to the compound non-standard marginal contamination effect.}
    \label{fig:misspec_lower_detail}
\end{figure}

\end{document}